\documentclass[aps,nofootinbib,preprintnumbers,citeautoscript]{revtex4}
\usepackage{enumerate}
\usepackage{graphicx}
\usepackage{amsthm}
\usepackage{mathtools}
\usepackage{textcomp}
\usepackage{braket}
\usepackage{booktabs}
\usepackage{longtable}
\usepackage{amsmath}
\usepackage{enumitem}
\usepackage{lscape}
\usepackage{amsfonts}
\usepackage{pifont}
\newcommand{\cmark}{\ding{51}}%
\newcommand{\xmark}{\ding{55}}%
\usepackage[export]{adjustbox}

  \makeatletter
  \def\ps@pprintTitle{%
  	\let\@oddhead\@empty
  	\let\@evenhead\@empty
  	\def\@oddfoot{}%
  	\let\@evenfoot\@oddfoot}
  \makeatother
  \renewcommand{\thetable}{\arabic{table}}

\begin{document}
\title{Quantum finite automata: survey, status and research directions}

\author{Amandeep Singh Bhatia$^\ast$ and  Ajay Kumar \\
\textit{Department of Computer Science, Thapar Institute of Engineering \& Technology, India} \\
E-mail: $^\ast$amandeepbhatia.singh@gmail.com}

\begin{abstract}
Quantum computing is concerned with computer technology based on the principles of quantum mechanics, with operations performed at the quantum level. Quantum computational models make it possible to analyze the resources required for computations. Quantum automata can be classified thusly: quantum finite automata, quantum sequential machine, quantum pushdown automata, quantum Turing machine and orthomodular lattice-valued automata. These models are useful for determining the expressive power and boundaries of various computational features. 
In light of the current state of quantum computation theory research, a systematic review of the literature seems timely. This article seeks to provide a comprehensive and systematic analysis of quantum finite automata models, quantum finite automata models with density operators and quantum finite automata models with classical states, interactive proof systems, quantum communication complexity and query complexity as described in the literature. The statistics of quantum finite automata related papers are shown and open problems are identified for more advanced research. The current status of quantum automata theory is distributed into various categories. This research work also highlights the previous research, current status and future directions of quantum automata models.

\end{abstract}
\maketitle



\theoremstyle{plain}

\newtheorem{thm}{Theorem}

\theoremstyle{definition}
\newtheorem{defn}{Definition}
\newtheorem{exmp}{Example}

\section{Introduction and motivation}

Quantum computing is a winsome field that deals with theoretical computational systems (i.e., quantum computers) combining visionary ideas of Computer Science, Physics, and Mathematics. It concerns with the behaviour and nature of energy at the quantum level to improve the efficiency of computations. Quantum computing relies upon the quantum phenomena of entanglement and superposition to perform operations. 

Feynman \cite{1} initially proposed the idea of quantum computing in 1982 after performing a quantum mechanics simulation on a classical computer. Up until then, quantum computing was thought to be only a theoretical possibility, but research over the last three decades has evolved such as to make quantum computing applications a realistic possibility. In 1994, Shor \cite{2} designed a quantum algorithm for calculating the factor of a large number $\it n$ with space complexity $O(log n)$ and time complexity $O((log n)^2 \ast log log n)$  on a quantum computer, and then perform $O(log n)$  post processing time on a classical computer, which could be applied in cracking various cryptosystems, such as RSA algorithm and elliptic curve cryptography. Through the impetus provided by Shor's algorithm, quantum computational complexity is an exhilarating area that transcends the boundaries of quantum physics and theoretical computer science. 
Theoretical research into quantum computing, along with experimental efforts to construct a quantum computer, has gained a lot of attention. In 1994, Shor \cite{2} introduced the concept of first quantum error correction code by representing the information of one qubit into a highly entangled state of nine qubits. In 1994, Wineland \cite{3} introduced the recognition of Controlled Notgate, using the two lowest energy levels of ions to realize the concept of two qubit states. In 1996, Grover \cite{4} designed an algorithm for finding an element in an unstructured set of size $\it n$   in $\sqrt{n}$ operations approximately. In 1998, Kwiat et al. \cite{5} implemented Grover's algorithm at Los Alamos National Laboratory using conventional optical interferometers.

Quantum finite automata blend quantum mechanics with classical finite automata. It is a theoretical model with finite memory for quantum computers, which plays a vital role in performing real-time computations. The theory of quantum automata has been developed using the principles of classical automata and quantum mechanics. Quantum automata lay down the vision of quantum processor for performing the quantum actions on reading the inputs.

In 1961, Landauer \cite{6} articulated a concept of reversibility in quantum computing. In 1985, Deutsch \cite{7} described a quantum Turing machine and determined the certainty of a universal quantum Turing machine based on the Church-Turing-Deutsch principle.  Furthermore, in 1993, Yao \cite{8} determined that if a particular function existed (computable in polynomial time using a quantum Turing machine), then a polynomial-size quantum circuit could be designed for that function. The field of quantum computation and information processing has subsequently made a significant impact on the academic and research community alike. 

The concept of quantum automata was first proposed by Moore and Crutchfield \cite{9} and Kondacs and Watrous \cite{10} independently. In 1997, Kondacs and Watrous [10] proposed a variant of quantum automata: $\textit{measure-many one-way quantum finite automata}$ (MM-1QFA). In 2000, Moore and Crutchfield \cite{9} proposed another variant of the quantum model: $\textit{measure-once one-way quantum finite automata}$ (MO-1QFA). MO-1QFA produces the output accept or reject after reading the last symbol of an input string; whereas MM-1QFA results in the output reject, accept, or continuation after reading each symbol of the input tape.

This article concentrates on the extensive survey of various quantum finite models. We summarized the existing literature in the form of systematic evolution of various models. A survey on closure properties and the equivalence of QFA models is also conducted. Further, we have shown bibliographic view of quantum finite automata and listed some open problems concerning quantum finite automata in which research can be carried out. 

\subsection{Motivation for research}
\begin{itemize}
	\item Quantum finite automata play a crucial role in quantum information processing theory. Investigation of the power of quantum finite automata is a natural goal. Therefore, this study focused on brief research on various models of quantum finite automata and explored to various directions.
	
	\item We recognized the requirement of comprehensive literature survey after considering progressive research in quantum automata theory. Therefore, we summarized the existing research based on wide and systematic search in this field and presented the research challenges for advanced research.
\end{itemize}
\subsection{Our Contributions}
\begin{itemize}
	\item A comprehensive investigation has been conducted to study various quantum finite automata models, quantum finite automata models with density operators and quantum finite automata models with classical states.	
	\item The aforementioned quantum finite automata models have been compared and categorized based on the closure properties, language recognition power and inclusive relation shown between them. Further, inclusion relationship among different models has been shown on basis of language recognition capability. 
	
	\item Statistical results (yearly publications, top cited papers, list of author’s publications) related QFA papers have been shown.
	\item Future research directions relating quantum finite automata models are presented.
\end{itemize}
\subsection{Related surveys}
Earlier surveys by Qiu et al. \cite{12}, Qiu and Li \cite{50}, Ambainis and Yakaryilmaz \cite{51} have been very innovative, but as the research have consistently grown in the field of quantum finite automata theory, there is a need for a methodical literature survey to evaluate, upgrade, and integrate the existing research presented in this field. Qiu and Li \cite{50} reviewed 
loosely some fundamental models of quantum computing
models from MO-1QFA to quantum Turing machine and partially outlined their
definitions, basic properties and mutual relationships. Ambainis and Yakaryilmaz \cite{51} provided a detailed survey and discussed the new research directions. Although, the data in http://publication.wikia.com/wiki/Quantum automata is not comprehensive and listed the publications till 2014. This research augments the previous surveys and presents a recent methodical literature survey to evaluate and discover the research challenges based on available existing research in the field of quantum automata theory.

\subsection{Review process}
The stages of this literature review include formal definitions of QFA models, comparison on the basis of language recognition power, inclusive relation between them, bibliometric perspective of QFA, investigating the closure properties, equivalence and minimization of QFA models, and exploration of research challenges. We posed the research questions shown in Table 1, and these questions helped us to collect the necessary information from papers in our review process.

\begin{table} [h]
	\centering
	\caption{Research questions and motivation}
	\begin{tabular}{ p{.01\textwidth} p{.35\textwidth}  p{.45\textwidth}}
		\hline
		1. & What is the current status of quantum automata theory? & Various quantum automata have been introduced. Among restricted 1QFAs, LAQFAs earn special attention. 2QFA models have been not usually inspected as compared to 1QFA models. Various algorithms have been developed to empower the measures of QFAs. The research challenge in terms of the research question is discovering the existing research that assessed and compared various QFA models. This study outlined their definitions, computational power, comparison, closure properties, and inclusive relation with other models; various types of existing research have been presented.\\
		2. & How to characterize classes of languages accepted by QFA models? & It provides the knowledge about the review done in this research article. It is mandatory to find out the number\\
		3. & How can we identify the relation between various QFA models? & of research papers in each type of quantum finite automata model, which helps to find the key research areas.\\
		4. & How to identify and classify the closure properties? &  The inclusive relation between QFA models is described on the basis of their language recognition capability. Closure properties of the languages accepted by various QFAs have been investigated. It has become the hotspot area in the quantum information processing. The latest research in quantum automata theory is going toward its connection with algebra and using it to study the power of QFA models. The research challenges in terms of research questions emphasize identifying the present prominence of research in QFA. Different research questions are used to identify the key research areas for future investigation in the field of quantum finite automata theory.\\
		5. & What are the new hot topics for further research in it? & Presently, developing the quantum interactive proof system, the connection between algebra and QFA models, promise problems are hot topics in quantum automata theory.\\
		6. & What is the percentage of publications every year? & Identify the number of papers published related quantum finite automata models until now.\\
		7. & What are the most highly cited papers? & Identify the papers which are highly influential and cited.\\
		8. & Which are the researchers having maximum number of publications relating quantum automata theory? & There are lots of research papers on all sorts of quantum finite automata from renowned research authors. Identify the number of publications of papers by them.\\
		\hline
	\end{tabular}
\end{table}

\subsection{Paper organization}
The organization of rest of this paper is as follows: Section 2 is devoted to preliminaries. In Section 3, a bibliometric perspective of QFA is presented.  In Section 4, we review one-way, 1.5-way, two-way QFA models, one-way and two-way quantum automata with classical states, one-way general QFA. Comparative studies and inclusion relation of these automata models are carried out in sub-section 4.10. Sections 5 present the statistical results of quantum automata related papers, literature survey including closure properties and equivalence of QFA models. In Section 6, variety of open problems concerning QFA models are presented followed by conclusion in Section 7. Note, a glossary of acronyms used in this paper can be found in “Appendix 1”.

\section{Preliminaries }
Before we start our tour, some preliminaries are given in this section. Linear algebra is an essential mathematical tool for quantum mechanics. Linear operators allow us to represent quantum mechanical operators as matrices and wave functions as vectors on some linear vector space. We assume that the reader is familiar with the notation of quantum mechanics; otherwise, reader can refer to quantum computational, quantum computing \cite{12, 62} and classical automata theory \cite{63}. Following, concepts of linear algebra are used in quantum automata theory: 
\begin{itemize}
	\item Linear vector space \cite{12}: It is defined as a set of elements, called vectors. It is closed under addition and multiplication by scalars. If two vector $\ket{\psi}$ and $\ket{\varphi}$ are a part of vector space, then $\ket{\psi}+ \ket{\varphi}$ belongs to vector space. There is also an operation of multiplication by scalars such that if $\ket{\psi}$ is in vector space, then $\alpha\ket{\psi}$ is in the space, where $\alpha$ is a complex scalar. 
	\item Bra-ket notation \cite{13}: In quantum mechanics, the bra-ket notation is a criterion for unfolding quantum states, composed of angle brackets and vertical bars.
	\begin{equation}
	\ket{u}=\begin{bmatrix}
	a_1\\
	a_2\\
	a_3\\
	\end{bmatrix}, \bra{v}=\begin{bmatrix}
	a_1^* &  a_2^* & a_3^*
	\end{bmatrix}, \ket{u}\bra{v}=\begin{bmatrix}
	a_1a_1^* &  a_1a_2^*  & a_1a_3^*\\
	a_2a_1^* &  a_2a_2^*  & a_2a_3^*\\
	a_3a_1^* &  a_3a_2^*  & a_3a_3^*\\
	\end{bmatrix} \end{equation}
	
	$a_i^*$ denotes the complex conjugate of the complex number $a_i$. The ket $\ket{u}$ is a column vector, and its conjugate transpose bra $\bra{v}$ is a row vector. It is also known as Dirac notation.
	\item Qubit: It is defined as a superposition of basis states. It can be represented as a state vector having two basis states labeled $\ket{0}$ and $\ket{1}$. In general, 
	\begin{equation}
	\ket{\psi}=\alpha\ket{0}+\beta\ket{1} \end{equation}
	where $\alpha$ and $\beta$ are complex numbers. The state $\ket{0}$ occurs with probability $|\alpha|^2$ and $\ket{1}$ with probability $|\beta|^2$. Since the absolute squares of the amplitudes equate to probabilities, it follows that $|\alpha|^2+ |\beta|^2=1$. One qubit represents two complex amplitudes ($\alpha$ and $\beta$ ), similarly, $\it n$ qubits represent $2^n$ complex amplitudes.
	\begin{figure}[h]
		\centering
		\includegraphics[scale=0.55]{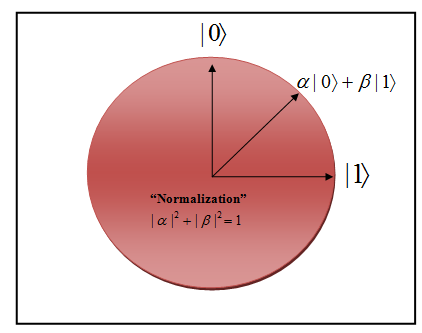}
		\caption{ Normalization of two-level quantum systems \cite{11}}
	\end{figure}
	\item Quantum state \cite{13}: A quantum state $\ket{\phi}$ is a superposition of classical states, 
	\begin{equation}
	\ket{\phi}=\alpha_1\ket{x_1}+\alpha_2\ket{x_2}+...+\alpha_n\ket{x_n} \end{equation}
	where $\ket{x_i}'s$ are classical states for $1 \leq i\leq n$, $\alpha_i's$ are complex numbers called amplitudes and $|\alpha_1|^2+ |\alpha_2|^2+...+|\alpha_n|^2=1$, where $|\alpha_i|^2$ is the squared norm of the corresponding amplitude.  Quantum state can also be seen as $\it n$-dimensional column vector.
	\begin{equation}
	\begin{bmatrix}
	\alpha_1\\
	\alpha_2\\
	...\\
	\alpha_3\\
	\end{bmatrix}
	\end{equation}
	\item Hilbert space: A Hilbert space $\it H$ is a complex vector space. It is a mathematical framework for describing the principles of the quantum system \cite{13} An inner product space on Hilbert space is associated with the inner product of vectors $\braket{u|v}: H\times H\rightarrow C$  satisfying the following properties for any vectors, such that $u,v,w \in H$  and $x,y \in C$  (a set of complex numbers).
	\begin{itemize}
		\item Linearity: $(x\bra{u}+y\bra{v})\ket{w}=x\bra{u}\ket{w}+y\bra{v}\ket{w}$
		\item Symmetric property: $\bra{u}\ket{v}=\bra{v}\ket{u}$
		\item Positive definite property: For any $u \in H, \bra{u}\ket{u}\geq0$ and $\bra{u}\ket{u}=0$ iff $u=0$.
	\end{itemize}
	\item Density matrix \cite{12}: It is an alternate representation of a state in quantum mechanics. Quantum mechanical systems can be in states which cannot be described by wave functions. Such states are called the mixed states and can be represented as a summation of orthonormal bases  $\ket{\Psi}'s$, $\rho=\sum_{i} \rho_i \ket{\Psi} \bra{\Psi}$, where $\rho_i$ is the probability for the system in the state of $\Psi_i$, and $\Psi_i's$ are the diagonal basis for $\rho$. $\rho_i's$ are called eigenvalues of the density matrix $\rho$. The density operator on Hilbert space must satisfy the trace condition, i.e. $\{Tr(\rho)=1|\rho\geq0\}$, where $\it Tr()$ refers to the sum of the diagonal elements of matrices.
	\item Vector norm: A vector norm $(\lVert v\rVert)$  is defined as a maupping from $\it R^n$  to  $\it R$ (Euclidean space) with the following properties:
	\begin{itemize}
		\item $\lVert v\rVert>0$, if $v\neq0$
		\item $\lVert \alpha v\rVert=|\alpha|\lVert v\rVert$, for any $\alpha \in R$
		\item $\lVert u+v\rVert\leq\lVert u\rVert+ \lVert v\rVert $, for any $u,v\in R^n$
	\end{itemize}
	\item Orthogonal and orthonormal vectors \cite{13}: Two vectors $\ket{u}$ and $\ket{v}$ are orthogonal, if they are perpendicular to each other i.e. the inner product of vectors $\bra{u}\ket{v}=0$. It can be defined as a set of vectors $V=\{v_1,v_2,... v_n\}$  are mutually orthogonal if every pair of vectors are orthogonal, i.e. $\bra{v_i}\ket{v_j}=0$, for all $i\neq j$. A set of vectors $\it V$ is orthonormal if every vector in $\it V$ is a unit vector and the set of vectors are mutually orthogonal. 
	\item Unitary evolution: In quantum systems, Markov matrices are replaced by matrices with complex number entries for the time evaluation of probabilistic systems, by maintaining the condition $\sum_{i=1}^{n} |\alpha_i|^2$.  Therefore, consider a quantum system state at time $t_0$: $\ket{\phi(t_0)}=\alpha_1\ket{x_1}+\alpha_2\ket{x_2}+...+\alpha_n\ket{x_n}$ change into the state at time $\it t$: $\ket{\phi'(t)}=\alpha_1'\ket{x_1}+\alpha_2'\ket{x_2}+...+\alpha_n'\ket{x_n}$ where amplitudes $\alpha_1,\alpha_2,...,\alpha_n$  and $\alpha_1',\alpha_2',...,\alpha_n'$ are related by $\ket{\phi'(t)}=U(t-t_0)\ket{\phi(t_0)}$, where $U(t-t_0)$is a time dependent unitary operator such that $(U(t-t_0))^*U(t-t_0)=1$, $\alpha_{ij}$'s are its entries for $1 \leq i,j \leq n$
	\begin{equation}
	\begin{bmatrix}
	\alpha_{11}& \alpha_{12} & ... & \alpha_{1n}\\
	\alpha_{21}&  \alpha_{22} & ... & \alpha_{2n}\\
	... &... & ...&\\
	\alpha_{n1}&  \alpha_{n2} & ... & \alpha_{nn}\\
	\end{bmatrix} \begin{bmatrix}
	\alpha_1\\
	\alpha_2\\
	...\\
	\alpha_n\\
	\end{bmatrix}=\begin{bmatrix}
	\alpha_1'\\
	\alpha_2'\\
	...\\
	\alpha_n'\\
	\end{bmatrix}
	\end{equation}
	and $\sum_{i=1}^{n} |\alpha_i|^2=\sum_{i=1}^{n} |\alpha_i'|^2=1$. Therefore, evaluation of a quantum system at any time must be unitary \cite{13}.
	\item Language recognition with unbounded and bounded error: A language $\it L$ is said to be accepted by quantum finite automata with cut-point $\lambda$ if $\forall x \in L$, the probability of acceptance is greater than $\lambda$ and  $\forall x \notin L$   the probability of acceptance is at most $\lambda$. A language is said to be accepted with bounded error if there exists $\varepsilon>0$ such that $\forall x \in L$, the probability of acceptance is greater than $\lambda+\varepsilon$ and $\forall x \notin L$  the probability of acceptance is less than $\lambda-\varepsilon$   A language $\it L$ is accepted by quantum finite automata with cut-point $\lambda$ and without a bounded error, then $\it L$  is accepted with an unbounded error \cite{12}. 	
\end{itemize}

\section{A Bibliometric Perspective of QFA}
Sixty-two Journal papers, 23 Conference papers and 3 Technical reports have been evaluated in this review systematically. Fig. 2 shows the number of papers published from year 1997 to 2018 in field of Quantum finite automata, where publication year on the x-axis and the number of papers published on the y-axis for papers in review. 

Sixty percentage of papers are journal papers, 29\% of papers are conference proceedings, 5\% are workshop papers, 4\% of papers are technical reports and 2\% are Thesis. 

\begin{longtable}{c l c} 
	\caption{List of sources publishing the top 15 quantities of articles of quantum automata}
	\label{variability_impl_mech}
	\endfirsthead
	\endhead
	\hline
	\textbf{Ranking} & \textbf{Publication source} & \textbf{Quantity} \\ 
	\hline
	1 & Theoretical Computer Science
	& 38 \\ 
	2 & Lecture Notes In Computer Science & 28 \\
	3 & International Journal of Foundations of Computer Science & 17\\
	4 & International Journal of Theoretical Physics & 11 \\
	5 & Information And Computation & 6\\
	6 & Journal of Computer And System Sciences & 4\\ 
	7 & Fundamenta Informaticae & 5 \\
	8 & Rairo Theoretical Informatics And Applications & 5 \\
	9 & Natural Computing & 5 \\
	10 & Quantum Information Computation & 5\\
	11 & Journal of Statistical Physics & 4\\
	
	12 &	Physical Review E & 4\\
	
	13 & Information Processing Letters & 4 \\
	
	14 & Physical Review A & 4 \\
	15 & Information Sciences & 3\\
	\hline
	
\end{longtable}
\begin{figure}[h]
	\centering
	\includegraphics[scale=0.65]{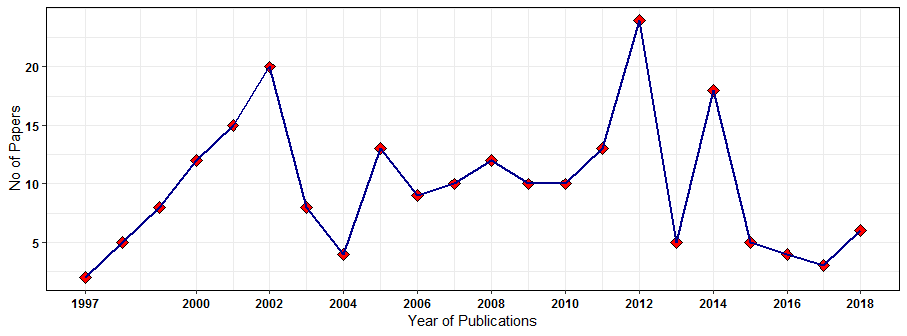}
	\caption{Number of papers per year in a review}
\end{figure}
Table 2 shows the publication sources which have published the top 15 quantities of contributions relating quantum finite automata. Fig 3 depicts the number of publication by various authors. There is a lot of literature on all sorts of quantum finite automata, particularly from the University of Latvia (Ambainis, Freivalds, Yakaryilmaz, etc). Table 3 depicts the top 15 highly influential and cited paper related quantum automata (source: https://scholar.google.co.in, accessed Jan 8, 2019).

\begin{figure}[h!]
	
	\includegraphics[scale=0.7]{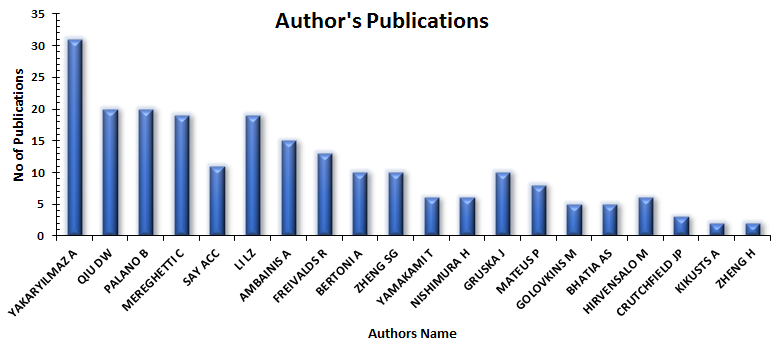}
	\caption{Author's publications}
\end{figure} 

\begin{longtable}{{p{1cm} p{6cm} p{5.5cm} p{2.3cm}}}
	\caption{Top 15 cited papers}
	\label{variability_impl_mech}
\endfirsthead
\endhead
\hline
	\textbf{Year} & \textbf{Title} & \textbf{Journal/
		Proceedings} & \textbf{Number of citations} \\ 
	\hline
	2000 & Quantum automata and quantum grammars\cite{9} & Theoretical Computer Science & 430\\
	1997 & On the power of quantum finite state automata \cite{10}& Foundations of Computer Science & 429 \\ 
	1998 &	1-way quantum finite automata: Strengths, weaknesses and generalizations \cite{15} &	Foundations of Computer Science &	307\\
	2002 &	Characterizations of 1-way quantum finite automata \cite{17} &	SIAM journal on computing &	182 \\
	2002&	Dense quantum coding and quantum finite automata \cite{53} &	Journal of the ACM	& 193 \\
	2000 &	Two-way finite automata with quantum and classical states \cite{24} &	Theoretical Computer Science &	176 \\
	2003 &	Quantum computing: 1-way quantum automata \cite{30} &	Developments in Language Theory &	91\\
	1999 &	Undecidability on quantum finite automata \cite{14} & ACM Symposium on Theory of Computing & 69\\
	2000 &	On the class of languages recognizable by 1-way quantum finite automata \cite{16} & Theoretical Aspects of Computer Science & 59\\
	2006 &	Algebraic results on quantum automata \cite{20} &	Theory of computing systems &	58\\
	2010 &	Succinctness of Two-way probabilistic and quantum finite automata \cite{26} &	Discrete Mathematics and Theoretical Computer Science &	60\\
	1999 & Probabilities to accept languages by quantum finite automata \cite{95} & International Computing and Combinatorics Conference & 50\\
	2012 & One-way finite automata with quantum and classical states \cite{92} & Languages alive (lecture notes in computer science) & 31\\
	2013 & State succinctness of two-way finite automata with quantum and classical states \cite{82} & Theoretical Computer Science &	25 \\
	2014 & On the state complexity of semi-quantum finite automata \cite{85} & RAIRO-Theoretical Informatics and Applications & 25 \\
	\hline		
\end{longtable}
Table 4 lists the Journals and Conferences publishing quantum finite automata related research, where (J, Journal; C, Conference; S, Symposium; W, Workshop; N, number of studies reporting related research as prime study; \#, total number of articles investigated). We observed that conferences like IEEE Symposium on Foundations of Computer Science (FOCS), ACM Symposium on Theory of Computing (STOC), International Conference on Computing and Combinatorics Conference (COCOON), Quantum Computation and Learning Workshop contribute a large part of research articles. Premier journals like Theoretical Computer Science, Journal of Theoretical Physics, Journal of Foundations of Computer Science, RAIRO: Theoretical Informatics and Applications contributed significantly to our review area.

\begin{longtable}{p{10.5cm} p{2cm} p{1cm} p{1cm}}
	\caption{Journals/Conferences reporting most related research}
	\label{variability_impl_mech}
	\endfirsthead
	\endhead
	\hline
	\textbf{Publication source} & \textbf{J/C/S/W} & \textbf{\#} & \textbf{N} \\ 
	\hline
	IEEE International Symposium on Foundations of Computer Science (FOCS) &	S &	4 &	4 \\
	Theoretical Computer Science &	J &	23 &	15\\
	ACM Symposium on Theory of Computing (STOC)&	S &	3 &	3\\
	International Computing and Combinatorics Conference (COCOON) &	C &	5 &	3 \\
	Journal of Theoretical Physics &	J &	7 &	2 \\
	Journal of Foundations of Computer Science & 	J &	5 &	4\\
	Proceeding of Quantum Computation and Learning Workshop &	W &	4 &	2 \\
	RAIRO: Theoretical Informatics and Applications &	J&	7 &	7\\
	Symposium on Fundamentals of Computation &	S &	3 &	1\\
	SIAM Journal on Computing  &	J &	2 &	1 \\
	Quantum Information and Computation	 & J &	3 &	3 \\
	International Conference on Developments in Language Theory (DLT) &	C &	4 &	4\\
	Symposium on Fundamentals of Computer Theory (FCT, Springer) &	S &	3 &	1 \\
	Journal of Natural Computing &	J &	5 &	2 \\
	Information Processing Letters &	J &	5 &	5\\
	International Symposium on Algorithms and Computation (ISAAC, Springer) &	S &	1 &	1 \\
	arXiv: preprint &	J &	18 & 8\\	
	\hline	
\end{longtable}
\section{Quantum Computing Models}
A quantum finite automaton is a quantum counterpart of a classical finite automaton. In quantum finite automaton, quantum actions are performed on reading the symbols from the inputs tape. Quantum finite automata can be classified into one-way quantum finite automata, 1.5-way quantum finite automata, two-way quantum finite automata, quantum sequential machines, quantum pushdown automata, quantum Turing machine and orthomodular lattice-valued automata. These automata act as a model of quantum processors. In this paper, we address the quantum finite automata models with finite memory \cite{9, 10, 14, 20, 30, 40}; quantum automata models with density operators \cite{18, 41, 42, 43} and quantum automata with both quantum and classical states \cite{24, 44, 92}.

\begin{defn} \cite{9}
	A Quantum finite automaton is defined as real-time quantum automaton, where
	\begin{itemize}
		\item An input alphabet {\it $\Sigma$}, 
		\item Hilbert space ${\it H}$, an initial state vector $s_{init} \in H$ with $|s_{init}|^2=1$, 
		\item A subspace $H_{accept}\subset H$ and an operator $P_{acc}$ which project on it, 
		\item A unitary transition matrix $U_a$ for $\forall \in \sigma$  . 
	\end{itemize}
\end{defn}	

The quantum language recognized by quantum finite automaton as a function $f_{QFA}(w)=|s_{init}U_w P_{acc}|^2$, where $U_w=U_{w_1}U_{w_2}...U_{w_{|w|}}$. The process of computation of input string $\it w$ starts with the initial vector, apply the unitary matrix of each symbol and measure the probability by applying projection operator such that the resultant state is in subspace $H_{accept}$.  

\subsection{One-way quantum finite automata} 
\begin{figure}[h]
	\centering
	\includegraphics[scale=0.7]{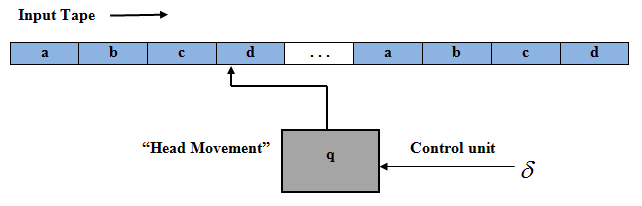}
	\caption{ One-way quantum finite automata}
\end{figure}
\begin{defn} \cite{11}
	A one-way quantum finite automaton (1QFA) is defined as a sextuple ${\it (Q,\Sigma, \delta, q_0, Q_{acc}, Q_{rej})}$, where 
	\begin{itemize}
		\item {\it Q} is a finite set of states, 
		\item {\it $\Sigma$} is an alphabet,
		\item $\delta$ is a transition function ${\delta}:Q\times \Sigma \times Q \rightarrow C$, where $\it C$ is a complex number.
		\item {\it $q_0$} is an initial state,
		\item $Q_{acc}\subset Q$ and $Q_{rej} \subset Q$ represent the set of accepting and rejecting states. 
	\end{itemize}
\end{defn}
In one-way QFA, the R/W head move only in the right direction for reading symbols from the input tape. During reading the input symbol from the input tape, it produces a superposition of states. Based on the measurement, one-way quantum finite automata are divided into measure-once and measure-many quantum finite automata. 

\begin{defn} \cite{9}
	MO-1QFA is defined as a quintuple ${\it (Q,\Sigma, \delta, q_0, Q_{acc})}$, where 
	\begin{itemize}
		\item {\it Q} is a set of states, 
		\item {\it $\Sigma$} is an input alphabet,
		\item {\it $q_0$} is a starting state,
		\item The transition function ${\it \delta}$ is defined by ${\it Q\times\Sigma\times Q \rightarrow C}$, it satisfy the unitary condition: 
		
		$$\sum \limits_{p\in Q}^{\forall (q_1,\sigma), (q_2,\sigma)\in Q\times\Sigma}	
		\overline{\delta(q_1,\sigma,p)}\delta(q_2,\sigma,p)=
		\left\{
		\begin{array}{ll}
		
		1~ q_1=q_2\\
		0~ q_1\neq q_2\\
		
		\end{array} \right \}$$

		\item ${\it Q_{acc}}$ is a set of accepting states.
	\end{itemize}
	
\end{defn}
The computation procedure of MO-1QFA consists of an input string {\it x=$\sigma_1\sigma_2...\sigma_n$}. For each symbol,  transition function is represented as unitary matrix $U(\sigma)$ such that $U(\sigma)(i,j)=\delta(q_j,\sigma,q_i)$  The R/W head reads {\it x} symbol by symbol from left to right side, and unitary matrices $U(\sigma_1),U(\sigma_2),...,U (\sigma_n)$ of each symbol are performed on the current state, beginning with $q_0$. $\delta$ is indicated by set of unitary matrices  $\{V_\sigma\}_{\sigma \in \Sigma}$, where $V_\sigma$ is a unitary evolution of MO-1QFA, defined as: $V_\sigma(\ket{q})=\sum\limits_{q'\in Q}{(q,\sigma,q')} \ket{q'}$. At the end, projective measurement is executed by projection operator on the final state in order to check whether the input string is accepted or rejected. It allows the measurement to be made only after reading the last symbol of the input string. If after reading the last symbol, a superposition $\psi=\sum\limits_{q_i\in Q_{acc}} \alpha_i \ket{q_i}+\sum\limits_{q_j\in Q_{rej}} \beta_j \ket{q_j}$ is observed, then the input string is accepted with $\sum {\alpha_i^2}$ and rejected with $\sum {\beta_j^2}$. The probability of acceptance for MO-1QFA is calculated as
$$P(x)=\lVert P_{acc}U(\sigma_n)...U(\sigma_2)U(\sigma_1)\ket{q_0}\lVert^2$$

\begin{defn} \cite{17}
	MM-1QFA is defined as a sextuple ${\it (Q,\Sigma, \delta, q_0, Q_{acc}, Q_{rej})}$, where 
	\begin{itemize}
		\item {\it Q} is a finite set of states, $Q=Q_{acc}\bigcup Q_{rej}\bigcup Q_{non}$, where ${Q_{acc},Q_{rej}, Q_{non}}$ denotes the accepting, rejecting and non-halting set of states states correspondingly. 
			\item {\it $\Sigma$} is an input alphabet,
			\item {\it $q_0$} is a starting state,
			\item The transition function ${\it \delta}$ is defined by ${\it Q\times\Sigma \cup \{\#, \$\} \times Q \rightarrow C}$, which represents the amplitudes flows from one state to other after reading the symbol from the input tape. It must satisfy the unitary condition.
			
		\end{itemize}
		
	\end{defn}
	The computation process of MM-1QFA consists of an input string {\it x=$\#\sigma_1\sigma_2...\sigma_n\$$}. The R/W head reads {\it x} from left-end marker \# to right-end marker \$ and transition function corresponding to each symbol is performed. The whole Hilbert space is divided into three subspaces: $E_{non}=span\{\ket{q}:q\in Q_{non}\},E_{acc}=span\{\ket{q}:q\in Q_{acc}\}$ and $E_{rej}=span\{\ket{q}:q\in Q_{rej}\}$  correspondingly there are three projectors ${\it P_{non},P_{acc}, P_{rej}}$ on to the three subspaces. After each transition, MM-1QFA measures its state with reference to observable ${\it E_{non}\bigoplus E_{acc} \bigoplus E_{rej}}$. If the observed state is in $E_{acc}$ or $E_{rej}$  subspace, then MM-1QFA accepts or rejects the input string respectively; otherwise, the computation process continues. Therefore, after every step of measurement, the superposition of states ends with measured subspace. 
			Due to non-zero probability of halting of MM-1QFA, it is convenient to have a path of the aggregate rejecting and accepting probabilities. Hence, the state of automata is denoted as $(\ket{\phi}, P_{acc}, P_{rej})$, where ${\it P_{acc}, P_{rej}}$ are the aggregate probabilities of accepting and rejecting. The transition $\delta$ is defined as: $P_{non}\ket{\phi^{'}}, P_{acc}+\lVert P_{acc}\ket{\phi^{'}} \lVert^2, P_{rej}+\lVert P_{rej}\ket{\phi^{'}}$, where $\ket{\phi^{'}}=V_{\sigma}\ket{\phi}$. The probability of acceptance for MM-1QFA is calculated as $P(x)=\sum_{k=1}^{n+1}\lVert P_{acc}U(\sigma_k)\prod_{i=1}^{k-1}(P_{non}U(\sigma_i))\ket{q_0}\lVert^2$.

			\subsection{Latvian quantum finite automata}
			Ambainis et al. \cite{20} introduced a generalized version of MO-1QFA named Latvian quantum finite automata (LQFA). It works similarly as MO-1QFA except that the transition function $(\delta)$, is a combination of projective measurement and unitary matrix \cite{12}. This alteration increases the power of LQFA for acceptance of languages.
			\begin{defn} \cite{12}
				LQFA is defined by septuple ${\it (Q,\Sigma,\{A_\sigma\}, \{P_\sigma\}, q_0, Q_{acc}, Q_{rej})}$, where 
				\begin{itemize}
					\item {\it Q} is a finite non-empty set of states, 
					\item {\it $\Sigma$} is an input alphabet and $\Gamma=\Sigma$ \{\#, \$\}, where \$ and \# are right and left-end markers respectively,
					\item $A_\sigma$ denotes unitary matrices for each symbol, and $P_\sigma$ is defined as a set of orthogonal subspaces,
					\item {\it $q_0$} is an initial state,
					\item {\it $Q_{acc}$} and $Q_{rej}$ are the set of accepting and rejecting states,		
				\end{itemize}
			\end{defn}
			LQFA is closely related to probabilistic finite automata \cite{20}. LQFA \cite{12} accepts the language with a bounded error if the language is a Boolean combination of the form $G_0b_1G_1...b_kG_k$ where $b_i's$ and $G_i's$ represents the letter and group languages respectively. Therefore, it can identify a proper subset of regular languages with bounded error acceptance mode \cite{20}. The computing process of LQFA starts with an initial state $q_0$ for input string ${\it x=\#\sigma_1\sigma_2...\sigma_n\$}$. On reading each symbol $\sigma \in \Sigma$, unitary matrix and projective measurement is performed. On the basis of measurement $P_\$$ of right-end marker, the input string is said to be accepted or rejected. Therefore, $P_\$=E_{acc}\bigotimes E_{rej}(E_{acc}=span\{\ket{q}:q \in Q_{Acc})$ and $(E_{rej}=span\{\ket{q}:q \in Q_{rej})$.
			
				\begin{table} [h]
				\centering
				\caption{Comparison of MO-1QFA and MM-1QFA}
				\begin{tabular}{ |p{4.4cm}|p{5.5cm}|p{5.5cm}| }
					\hline
					\textbf{Model properties} 
					& {\textbf{MO-1QFA}} 
					& {\textbf{MM-1QFA}}\\
					
					\hline
					\textbf{Proposed by} & Moore and Crutchfield \cite{9}
					& Kondacs and Watrous \cite{10} \\ 	
					\hline
					\textbf{Computation process} & Measurement is allowed only once after reading the last symbol of the input string.   &  Measurement is allowed after reading of each symbol by R/W head from the input tape.  \\ 	\hline
					\textbf{Measurement result} & Accept/ reject. & 		
					Accept/ reject/ continuation. \\ 	\hline
					\textbf{Language acceptance} & It can accept only group languages. & It is strictly more powerful than MO-1QFA for bounded error acceptance.  \\
					\hline
				\end{tabular}
			\end{table}

\subsection{QFA with control language}
In this model, the measurement is performed after reading each symbol from the input tape and R/W head is allowed to move only in right direction of input tape. 
\begin{defn} \cite{30}
	CL-1QFA is defined as a quintuple ${\it (Q,\pi,\{W(\sigma)\}_{\sigma \in\Sigma}, O, L)}$, end-marker $\$ \notin \Sigma$ and $\Sigma\cup\{\$\}$, where
		\begin{itemize}
			\item {\it Q} is a finite set of states, 
			\item {\it $\pi\in C^{1\times n}$} is the initial amplitude satisfying  satisfying $\lVert \pi \rVert^2=1$,
			\item $U(\sigma)\in C^{n\times n}$ is a unitary matrix, 
			\item $\it O$ is an observable on $C^{1\times n}$, if $C=\{c_1,c_2,...,c_s\}$ is a set of all possible results of measurements of $\it O$ and $\{P(c_i):i=1,2,...n\}$  denotes the projector onto the eigenspace corresponding to $c_i$, for all $c_i \in C$,
			\item $L\subseteq C^\ast$ is a regular language (control language),		
		\end{itemize}
	\end{defn}
	CL-1QFA allows an arbitrary projective measurement on the Hilbert space spanned by $\it Q$. The computation procedure is different from MO-1QFA and MM-1QFA models. An observable $\it O$ is considered with set of possible results $C=\{c_1,c_2,...,c_s\}$.   On any given input string $\it x$, the computation displays a sequence $y \in C^\ast$ of measurement results with a certain probability $p(y\rvert x)$, the input string is accepted iff $\it y$ belongs to a fixed regular control language $L\subseteq C^\ast$.

	The computation procedure of CL-1QFA consists of two steps: Firstly, unitary matrix $U(\sigma)$ is applied to the current state $\ket{\phi}$ and produces a new state such that $\ket{\psi}=U(\sigma)\ket{\phi}$. Secondly, the resultant state is observed by an observable $\it O$, which produces a result $c_k$   with probability $P=\lVert(c_k)\ket{\psi}\rVert$, and state is collapsed to $\frac{p(c_k)\ket{\psi}}{\sqrt{P}}$. Therefore, the computation on input string $x_1,x_2,...x_n$ leads to a sequence $y_1y_2...y_{n+1} \in C^\ast$, where $x_{n+1}=\$$ with probability $P(y_1y_2...y_{n+1}\rvert x_1,x_2,...x_n\$)$  is calculated as $P(y_1y_2...y_{n+1}\rvert x_1,x_2,...x_n\$)=\lVert \prod_{i=1}^{n+1}P(y_{n+2-i})U(x_{n+2-i}) \ket{q_0}\rVert ^2$,  Thus, the probability of acceptance is calculated as 
	$$P(x)=\sum_{y_1y_2,...y_{n+1}}P(y_1y_2...y_{n+1} x_1,x_2,...x_n\$)$$.
	
	\subsection{Ancilla QFA}
	In QFA, each transition must be unitary, which limits its computational power. Therefore, Paschen \cite{40} introduced a different 1QFA with ancilla qubits (AQFA) in order to avoid the restriction of unitary transitions. It can be done by adding an output alphabet to the MO-1QFA in Definition 2.
	
	\begin{defn}
		An AQFA is defined as a septuple  ${\it (Q,\Sigma, \Omega, \delta, q_0, Q_{acc}, Q_{rej})}$, where 
		\begin{itemize}
			\item {\it Q} is a finite set of states, $Q=Q_{acc}\bigcup Q_{rej}\bigcup Q_{non}$, where ${\it Q_{acc},Q_{rej}, Q_{non}}$ denotes the accepting, rejecting and non-halting set of states states correspondingly. 
				\item {\it $\Sigma$} is an input alphabet,
				\item $\Omega$ is a output alphabet,
				\item {\it $q_0$} is a starting state,
				\item The transition function ${\it \delta}$ is defined by ${\it Q\times\Sigma\times Q \times \Omega \rightarrow C}$, it satisfy the unitary condition: 
				
				$$\sum \limits_{p\in Q,~\omega\in \Omega}^{\forall (q_1,\sigma), (q_2,\sigma)\in Q\times\Sigma}	
				\overline{\delta(q_1,\sigma,p,\omega)}\delta(q_2,\sigma,p,\omega)=
				\left\{
				\begin{array}{ll}
				
				1~ q_1=q_2\\
				0~ q_1\neq q_2\\
				
				\end{array} \right \}$$
			\end{itemize}
		\end{defn}
		\subsection{One-way general quantum finite automata}
		The measure-many version of LQFA was studied by Nayak \cite{41} and Ambainis et al. \cite{18}. This model is named as general quantum finite automata (GQFA). The definition is almost same as MM-1QFA except that for any symbol $\sigma$ in the input alphabet   induces a transition function, which is a combination of unitary transformation and projective measurement instead of a unitary transformation only. Li et al. \cite{43} studied the generalized version of 1QFA, called one-way general quantum finite automata (1gQFA), in which each symbol in the input alphabet induces a trace-preserving quantum operation instead of a unitary transformation. There are two types of 1gQFA: measure-once one-way general quantum finite automata (MO-1gQFA) and measure-many one-way general quantum finite automata (MM-1gQFA). 
		\subsubsection{Measure-once one-way general quantum finite automata}
		Hirvensalo \cite{42} introduced a model in which transition function corresponding to each input symbol induces a completely positive trace preserving mapping. It is a generalized version of MO-1QFA. 
		\begin{defn} \cite{43}
			MO-1gQFA is defined as a quintuple ${\it (H,\Sigma, \rho_0, \{\xi_\sigma\}_{\sigma \in \Sigma},P_{acc})}$, where 
			\begin{itemize}
				\item {\it H} is a finite- dimensional Hilbert space, 
				\item {\it $\Sigma$} is an input alphabet,
				\item $\rho_0$ is initial state (density operator on  $\it H$),
				\item $\xi$ is a trace-preserving quantum operation on $\it H$, for each $\sigma \in \Sigma$ ,
				\item $P_{acc}$ is projector called accepting subspace of $\it H$, $P_{rej}=1-P_{acc}$, then the $\{P_{rej},P_{acc}\}$  forms projective measurement on $\it H$.
			\end{itemize}
		\end{defn}
		Consider an input string $x=\sigma_1 \sigma_2...\sigma_n$   the computation of MO-1gQFA is progress as follows: firstly, the quantum operations $\xi_{\sigma1}\xi_{\sigma2},...\xi_{\sigma n}$  are performed on $\rho_0$  after reading the input symbols. At the end, projective measurement $\{P_{rej},P_{acc}\}$ is applied on the final state. On the basis of measurement, the input string is said to be accepted with certain probability. Therefore, it induces a function such that $f_{MO-1gQFA}: \Sigma^\ast \rightarrow [0,1]$ as $$ f_{MO-1gQFA}(x): Tr(P_{acc}\xi_n \circ ... \circ \xi_2 \circ \xi_1(\rho_0))$$
		where $\circ$  denotes the composition. Hence, $f_{MO-1gQFA}(x)$  defines the probability of acceptance for input string $\it x$.  
		
		\subsubsection{Measure-many one-way general quantum finite automata}
		Li et al. \cite{43} studied MM-1gQFA from the language recognition power and the equivalence problem. It has been proved that the class of languages recognized by MM-1gQFA is regular languages with bounded error. 
		\begin{defn} \cite{43}
			MM-1gQFA is a sextuple ${\it (H,\Sigma, \rho_0, \{\xi_\sigma\}_{\sigma \in \Sigma},H_{acc}, H_{rej})}$, where 
				\begin{itemize}
					\item {\it H} is a finite- dimensional Hilbert space, 
					\item {\it $\Sigma$} is an input alphabet,
					\item $\rho_0$ is initial state (density operator on  $\it H$),
					\item $\xi$ is a trace-preserving quantum operation on $\it H$, for each $\sigma \in \Sigma$,
					\item $H_{acc}, H_{rej}$ are the accepting and rejecting subspaces of $\it H$, respectively. Therefore, $\{H_{acc},H_{rej}, H_{non}\}$ spans the full space of $\it H$ .There is a measurement $\{P_{acc},P_{rej}, P_{non}\}$, of projectors onto a subspace $H_{acc}, H_{rej}$, and $H_{non}$ respectively.
		\end{itemize}
	\end{defn}
The input string $\it x$ is written on an input tape with both end-markers. The computation procedure of MM-1gQFA is same as that of MM-1QFA. Firstly, quantum operation is performed on current state $\rho$ Then, resultant state is measured using set of projectors $\{P_{acc},P_{rej},P_{non}\}$  If the observed state is in $H_{acc}$ or $H_{rej}$ subspace, then MM-1gQFA accepts or rejects the input string respectively; otherwise, the computation process continues with probability $Tr(P_{non} \xi_\sigma (\rho))$  In order to represent the total number of states of MM-1gQFA, we have defined $\nu:L(H)\times R\times R$. Thus, the current state of an automata is described as a triplet $\{\rho, P_{acc}, P_{rej}\} \in \nu$. The evolution of MM-1gQFA on reading a symbol $\sigma$  can be defined by an operator $T_\sigma$ on $\nu$ as  
$$ T_\sigma(\rho, P_{acc}, P_{rej})=(P_{non} \xi_\sigma(\rho)P_{non}, Tr(P_{acc} \xi_\sigma(\rho))+P_{acc}, Tr(P_{rej} \xi_\sigma(\rho))+P_{rej}) $$

\subsection{One-way quantum finite automata with quantum and classical states}
Qiu et al. \cite{44} introduced a new computing model of 1QFA named one-way quantum finite automata together with classical states (1QFAC). In this model, there are both quantum and classical states. It performs only one measurement for computing each input string i.e. after reading the last symbol. Measurement is performed according to the last classical state reached after processing the input string. 
\begin{defn} \cite{44}
A 1QFAC is defined as a nonuple ${\it (S, Q, \Sigma, \Gamma, s_0, \ket{\psi_0},\delta, U, M)}$, where 
\begin{itemize}
\item {\it S} is a finite set of classical states,  
\item {\it Q} is a finite set of quantum states,  
\item {\it $\Sigma$} is an input alphabet,

\item {\it $\Gamma$} is an output alphabet,
\item {\it $s_0$}  is an initial classical state,
\item $\ket{\psi_0}$ is a unit vector in Hilbert space $H(Q)$ (initial quantum state),
\item $\delta$ is a transition function: $S\times\Sigma \rightarrow S$, (the classical transition map),
\item $U=\{U_{s\sigma}\}_{s\in S, \sigma\in \Sigma}$, where $U_{s\sigma}:H(Q) \rightarrow H(Q)$ is a unitary operator for each $\it s$ and $\sigma$, 
\item  $M=\{M_s\}_{s\in S}$, where each $M_s$ is a projective measurement over $H(Q)$ with outcomes in $\Gamma$. 

\end{itemize}
\end{defn}
Therefore, each $M_s=\{P_s,\gamma\}_{\gamma \in\Gamma}$ such that $\sum_{\gamma\in \Gamma}p_{s,\gamma}=I$ and $P_{s,\gamma}P_{s,\gamma'}=	
\left\{
\begin{array}{ll}
P_{s,\gamma} ~~\gamma=\gamma'\\
0~~~~ \gamma\neq\gamma'\\
\end{array} \right \}$. After reading the input string, if the classical state is $it s$ and quantum state is in $\ket{\psi}$ then $\lVert P_{s,\gamma} \ket{\psi}\rVert^2$ is the probability of producing $\gamma$ as a result on input string. In above definition, $\Gamma=\{a,r\}$ where $\it a, r$ denotes acceptance and rejection of string respectively. Therefore, $M=\{P_{s,a}, P_{s,r},s \in S\}$, where $P_{s,a}, P_{s,r}$   are two projectors such that $P_{s,a}+P_{s,r}=I$ and $P_{s,a} P_{s,r}=0$.

The computing process of 1QFAC for an input string $x=\sigma_1,\sigma_2,...\sigma_n \in \Sigma^{\ast}$ is described as follows: It starts with an initial classical state $\it s_0$ and initial quantum state $\ket{\psi_0}$. On reading $\sigma_1$ the classical state changes into $\mu\sigma_1$ and quantum state becomes $U_{s_0\sigma_1} \ket{\psi_0}$. Further, on reading the next symbol, the classical states changes into $\mu(\sigma_1\sigma_2)$ and quantum state to the result of applying $U_{\mu(\sigma_1)\sigma_2}$ to $U_{s_0\sigma_1} \ket{\psi_0}$. This transformation of states occurs in succession till the last symbol. Thus, on reading the symbol $\sigma_n$ the classical state becomes $\mu(x)$ and quantum state is as $U_{\mu(\sigma_1...\sigma_{n-2}\sigma_{n-1}){\sigma_n}}U_{\mu(\sigma_1...\sigma_{n-3}\sigma_{n-2})}{\sigma_{n-1}}...U_{\mu(\sigma_1)\sigma_2}U_{s_0\sigma_1} \ket{\psi_0}$. Let $\mu(Q)$  is a set of unitary operators on $H(Q)$. For sake of accessibility, we define mapping $\it \nu:\Sigma^\ast \rightarrow \mu(Q)$ as: $\nu(\epsilon):I$, where $\it I$ denotes the identity operator on $H(Q)$.
$$ \nu(x)= U_{\mu(\sigma_1...\sigma_{n-2}\sigma_{n-1})_{\sigma_n}}U_{\mu(\sigma_1...\sigma_{n-3}\sigma_{n-2})_{\sigma_{n-1}}}...U_{\mu(\sigma_1)\sigma_2}U_{s_0\sigma_1} \ket{\psi_0} $$
Finally, the probability of 1QFAC on reading the input string $\it x$ produces a result $\gamma$ as: $$ P(x)=\lVert P_{\mu(x),\gamma}\nu(x) \ket{\psi_0}\rVert^2$$
\subsection{1.5-way quantum finite automata}
In 1.5-way quantum finite automata (1.5QFA), R/W head is allowed to remain stationary or move towards the right direction of the input tape, but it cannot move towards the left of input tape. Amano and Iwama \cite{14} proved that if the input tape is circular, then 1.5QFA can be designed for non-context-free languages. They have not considered the right-end marker ${\$}$ on the input tape. 
\begin{defn} \cite{39}
1.5QFA is defined by sextuple ${\it (Q,\Sigma, \delta, q_0, Q_{acc}, Q_{rej})}$, where 
\begin{itemize}

\item {\it Q} is a finite set of states,  
\item {\it $\Sigma$} is an input alphabet and $\Gamma=\Sigma \cup \{\#,\$\}$,	  
\item Transition function $(\delta)$ satisfies the condition: $\delta(q,\alpha,p,\leftarrow)=0$  for $\it p,q \in Q$, $\leftarrow$ is a head movement towards left direction (not allowed) and $\alpha \in \Gamma$
\item {\it $q_0$} is an initial classical state,
\item $Q_{acc},Q_{rej}$ represent the set of accepting and rejecting states.
\end{itemize}
\end{defn}

\subsection{Two-way quantum finite automata}
Two-way quantum finite automaton (2QFA) is a quantum variant of two-way finite automata. In 2QFA model, R/W head can remain stationary or move either in left or right direction. 2QFA is more dominant than the classical model. 
\begin{figure}[h]
	\centering
	\includegraphics[scale=0.6]{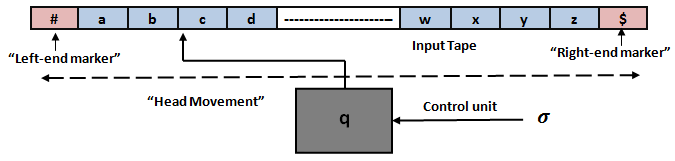}
	\caption{Two-way quantum finite automata}
\end{figure}
\begin{defn}
	Two-way quantum finite automaton is a sextuple $(Q, \Sigma, \delta,q_0, Q_{acc}, Q_{rej})$, where
	\begin{itemize}
		\item $\it Q$ is a set of states. Moreover, $Q=Q_{acc}\cup Q_{rej}\cup Q_{non}$, where $Q_{acc}, Q_{rej}, Q_{non}$ represent the set of accepting, rejecting and non-halting states respectively. 
		\item $\Sigma$ is an alphabet,
		\item Transition function $\delta$ is defined by $\delta: Q\times\Sigma \cup \{\#, \$\}\times Q\times D \rightarrow C$, where $D=\{\leftarrow,\uparrow,\rightarrow\}$ represent the left, stationary and right direction of read/write head. Transition function must satisfy the following conditions:
	\end{itemize}
		\begin{enumerate}[label=(\alph*)]
			\item Local probability and orthogonality condition:
			$$ \sum_{(q', d)\in Q\times D}^{\forall(q_1,\sigma_1),(q_2, \sigma_2)\in Q \times\Sigma} 	\overline{\delta(q_1,\sigma,q',d)}\delta(q_2,\sigma,q',d)=
			\left\{
			\begin{array}{ll}
			
			1~ q_1=q_2\\
			0~ q_1\neq q_2\\
			
			\end{array} \right \} $$
			\item First separability condition:
			$$ \sum_{q'\times Q}^{\forall(q_1,\sigma_1),(q_2, \sigma_2)\in Q \times\Sigma} 	\overline{\delta(q_1,\sigma_1,q',\rightarrow)}\delta(q_2,\sigma_2,q',\uparrow)+\overline{\delta(q_1,\sigma_1,q',\uparrow)}\delta(q_2,\sigma_2,q',\leftarrow)=0$$
			\item Second separability condition:
			$$ \sum_{q'\times Q}^{\forall(q_1,\sigma_1),(q_2, \sigma_2)\in Q \times\Sigma} 	\overline{\delta(q_1,\sigma_1,q',\rightarrow)}\delta(q_2,\sigma_2,q',\leftarrow)=0$$
		\end{enumerate}
\end{defn}
A 2QFA is simplified, for each $\sigma \in \Sigma$, if there exists a unitary linear operator $V_\sigma$ on the inner product space such that $L_2\{Q\}\rightarrow L_2\{Q\}$, where $\it Q$ is the set of states and a function $\it D:Q\rightarrow\{\leftarrow, \uparrow, \rightarrow\}$. Define transition function as
	\begin{equation}
\delta(q,\sigma,q',d)= \left\{ \begin{array}{l}
\bra{q'}{V_\sigma}\ket{q} \\
0 \end{array}
\middle\vert\;
\begin{array}{@{}l@{}}
\text{if} ~ D(q')=d \\
\text{else}
\end{array}
\right\}
\end{equation}
where $\bra{q'}{V_\sigma}\ket{q}$ is a coefficient of $\ket{q'}$ in  $V_\sigma \ket{q}$.

In order to process the input string by $M_{2QFA}$ , we assume that input string $\it x$ is written on input with both end-markers such that $\#x\$$. The automata is in any state $\it q$ R/W head is above the symbol $\sigma$. Then,  with the amplitude $\delta(q,\sigma,q',d)$ moves to state  $\it q',d\in \{\leftarrow, \uparrow, \rightarrow\}$, moves the R/W head one cell towards left, stationary and in right direction. The automata for processing an input $\it x$ corresponds a unitary evolution in the inner-product space $\it H_n$.

A computation of a 2QFA $M_{2QFA}$ is a sequence of superpositions $c_0,c_1,c_2,....,$ where $c_0$ is an initial configuration. When the automata are observed in a superposition state, for any $c_i$, it has the form $U_\delta\ket{c_i}\sum_{c\in C_n}\alpha_c\ket{c_i}$ where  defines the set of configurations, and the configuration $c_i$ is associated with amplitude $\alpha_c$  Superposition is valid; if the sum of the absolute squares of their probability amplitudes is unitary. The probability for a specified configuration is given by the absolute squares of amplitude associated with that configuration. Time evolution of quantum systems is given by unitary transformations. Each transition function $\delta$ induces a linear time evolution operator over the space $\it H_n$.
$$ U_{\delta}^x \ket{q,k}=\sum_{(q',d)\in Q\times D}\delta(q,x(k),q',d)\ket{q',k+d mod\lvert x \rvert}$$
for each $(q,k)\in C_{\lvert x\rvert}$, where $q\in Q, k\in Z_{\lvert x\rvert}$ and extended to $\it H_n$ by linearity \cite{10}. 
\subsection{Two-way Quantum Finite Automata with Quantum and Classical States}
Ambainis and Watrous \cite{24} introduced two-way finite automata with quantum and classical states (2QCFA). In this model, the internal state may be a (mixed) quantum state and the tape head position is classical. It is an intermediate model between 1QFA and 2QFA. 
\begin{defn} \cite{24}
A 2QCFA is defined as a nonuple $(S, Q, \Sigma, \Theta,  \delta,q_0, s_0, S_{acc}, S_{rej})$, where
\begin{itemize}
	\item $\it S$ is a finite set of classical states,
	\item $\it Q$ is a finite set of quantum states,
	\item $\Sigma$ is an input alphabet,
	\item $\Theta$ defines the evolution of the quantum portion of the internal state, 
	\item $\delta$ defines the evolution of the classical part,
	\item $q_0$is an initial quantum state $q_0 \in Q$,
	\item $s_0$ is an initial classical state $s_0 \in S$ ,
	\item,$S_{acc}, S_{rej}$ are the set of accepting and rejecting states respectively, $(S_{acc},S_{rej}\subseteq S)$  , 
	\item $\delta$ is a transition function $\delta$: $S\times \Sigma \rightarrow S$, (the classical transition map)
	
\end{itemize}
\end{defn}
Consider an input string $\it x$ the computation procedure of 2QCFA is as follows: Initially, tape squares are indexed with $1,2,...,\lvert x \rvert=n$ consists $x_1,x_2,...,x_n$ and squares are indexed with 0 to $\it n+1$ including right and left-end markers. On reading the symbol $\sigma$ the quantum state transformed according to $\Theta(s,\sigma)$ and further classical state and R/W head position changed according to $\delta(s,\sigma)$ and along with the outcome achieved from measurement of $\Theta(s,\sigma)$. Subsequently, the outcomes of each measurement are probabilistic, and the transitions for classical state may be also probabilistic. Therefore, for any input string, the 2QCFA is said to be accepted with probability  $S_{acc}(x)$  when the computation enters classical accepting state $S_{acc}$, otherwise rejected with probability $S_{rej}(x)$.Thus, the computation process is halted when it enters either classical accepting or rejecting state.

\subsection{Comparison of various quantum finite automata}

\begin{figure}[h]
	\centering
	\includegraphics[scale=0.38]{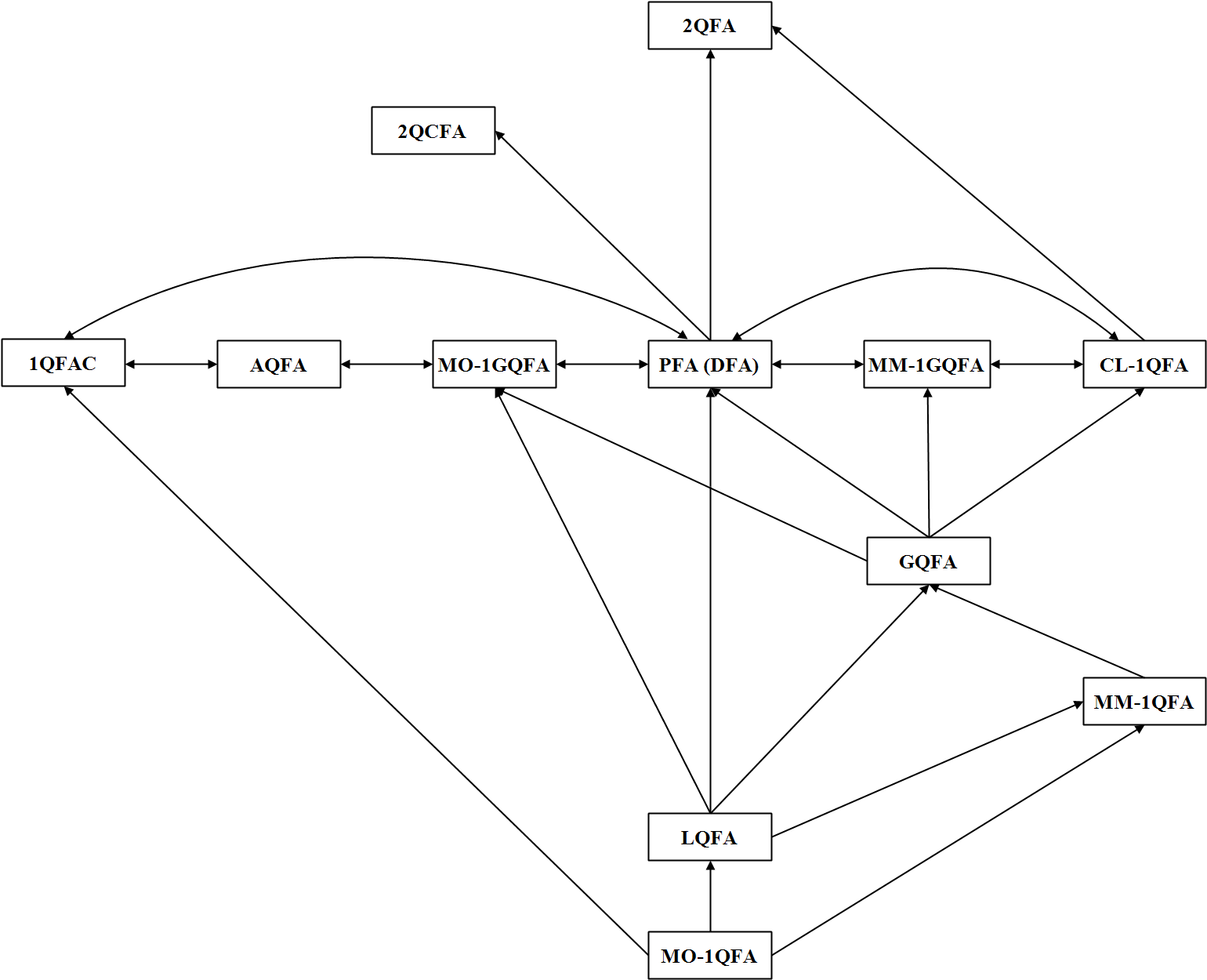}
	\caption{Inclusions among the languages recognized by QFA models with bounded error}
\end{figure}

In this section, we review comparative studies of various types of quantum finite automata. Fig 6 shows inclusion relationship among the languages recognized with a bounded error by 1QFA: MO-1QFA, MM-1QFA, LQFA, AQFA, CL-1QFA, 1QFAC; GQFA: MO-1gQFA, MM- 1gQFA, 2QFA and 2QCFA. The acronyms of models used to denote the classes of languages recognized by them, e.g. ‘‘MO-1gQFA’’ depicts the class of languages recognized by MO-1gQFA with bounder error. One-directional lines show containment relation and bidirectional lines show equivalence relation. The relationship among the model shows: The languages recognized by MO-1QFA are contained in those recognized by MM-1QFA. MM-1gQFA, CL-1QFA, MO-1gQFA, ancilla QFA recognize the same class of languages (i.e., regular languages). MO-1gQFA can simulate DFA and even probabilistic automata. Thus, both recognize exactly regular languages. 2QCFA is more powerful than two-way PFA because it can recognize regular languages with certainty and also some non-regular languages in polynomial time. MM-1gQFA, CL-1QFA, PFA, MO-1gQFA, AQFA, 1QFAC recognizes the same class of languages (i.e., regular languages). Table 6 depicts comparative study of various types of quantum finite automata.

\section{Literature Review}
Initial research into quantum automata was conducted by Moore et al. \cite{9}, Kondacs et al. \cite{10}, Ambainis et al. \cite{15,16}, and Brosky et al. \cite{17}. Thereafter, a significant body of research has been produced in the field of quantum automata theory. As such, a systematic review of the literature, in order to appraise, the current state of quantum automata research. Moore et al. \cite{9} introduced the concept of quantum finite and quantum pushdown automata. Their investigation revealed that quantum regular and context-free languages can be represented by quantum finite automata and quantum pushdown automata respectively.  Furthermore, they studied various closure properties, pumping lemmas and rational generating functions in parallel to classical automata theory. 
\begin{figure}[h]
	\centering
	\includegraphics[scale=0.55, frame]{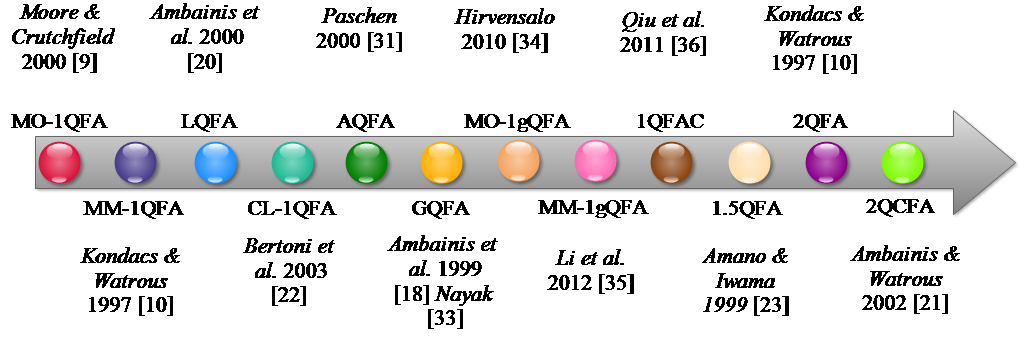}
	\caption{Summary of various quantum finite automata models}
\end{figure}

Brodsky et al. \cite{17} identified the language of acceptance, closure properties, and equivalence in 1QFA models. They designed MO-1QFA with bounded errors for group languages. Furthermore, they designed an algorithm for the equivalence of two MO-1QFAs and proved that probabilistic finite automata can simulate the MO-1QFA. Ambainis et al. \cite{16} demonstrated that languages recognized by 1QFA were not closed under Boolean or union operations. Furthermore, they analyzed the necessary and sufficient conditions for languages to be recognized by 1QFA. MM-1QFA, a variant of 1QFA is significantly more powerful than MO-1QFA for bounded error acceptance \cite{15}. In case of MM-1QFA, measurements are taken after reading each symbol from the input tape causing the string to be rejected without processing the complete input string. Brodsky et al. \cite{17} demonstrated that languages recognized by MM-1QFA are closed under word quotients, inverse homomorphisms, and complement operations, but are not closed under homomorphism.

Ambainis et al. \cite{15} demonstrated that MM-1QFA can accept a language with a probability of more than 7/9 if it is accepted by one-way reversible finite automata (1RFA). In fact, they showed that if we allow smaller probabilities, MM-1QFA can be more powerful than 1RFA. Furthermore, they demonstrated that MM-1QFA can recognize the language prime $L_p=\{a^p\mid where ~ p~ is ~ a~ prime\}$ with a probability close to 1, which is exponentially smaller than probabilistic finite automata (PFA). Ambainis et al. [18] demonstrated that MM-1QFA can be exponentially outsize than its corresponding deterministic finite automata for a particular language.  Kikusts \cite{19} designed an MM-1QFA for a language, which requires quadratically fewer states than its corresponding DFA.

Ambainis et al. \cite{20} introduced a generalized version of MO-1QFA: LQFA. In LQFA, the transition function is a combination of a unitary matrix and a projective measurement \cite{12}. They demonstrated that the class of languages recognized by LQFA are closed under complement, union, word quotients, and inverse homomorphism. LQFA with bounded errors can be designed for a proper subset of the languages for which MM-1QFA can be designed. Therefore, a MM-1QFA can be designed for all languages for which LQFA can be designed. Bertoni et al. \cite{99} explored 1QFA model where only measurements are allowed (MON -1QFA). It has been shown that the class of probabilistic behaviors of MON -1QFA is closed under under Hadamard product and \textit{f}-complement.

In 2003, Bertoni et al. \cite{30} proposed a quantum computing model named quantum finite automata with control states (CL-1QFA). CL-1QFA allows an arbitrary projective measurement on the Hilbert space spanned by $it Q$. It has been shown that the class of languages recognized by CL-1QFA with an isolated cut point is closed under Boolean operations, whereas MM-1QFA is not closed under Boolean operations \cite{16}. Furthermore, it has been proved that the languages recognized by CL-1QFA are regular languages with bounded error. Qiu and Yu \cite{103} studied the concept of multi-letter QFA ($QFA_*$) proposed by Belovs et al. \cite{104}. It has been shown that a language $L= ((a+b)^{*}b)$ can be recognized by 2-letter QFA with no error, which cannot be recognized by 1QFA. Qiu and Yu extended their work and proved that \textit{(k+1)}-letter QFA is more powerful than \textit{k}-letter QFA in language recognition. Qiu et al. \cite{101} studied the minimization decidability problems of the equivalence of multi-letter QFA.

Paschen \cite{40} introduced a new QFA model by adding some ancilla qubits to avoid the restriction of unitarity. This is done by addition of an output alphabet. Ciamarra \cite{45} introduced a new model of 1QFA having computational power is at least equal to classical automata. Ancilla QFA and Ciamarra QFA are special MO-1GQFA. Bhatia and Kumar \cite{112} proved that quantum Muller automaton is more dominant than quantum Büchi automaton. Further, the closure properties of quantum $\omega$-automata are proved.

It has been shown that GQFA \cite{18, 41} has exponentially more states than DFA for $L=\{w0 \mid w \in \{0,1\}^*\mid w \mid \leq m, m\geq 1\}$. Nayak \cite{41} showed that GQFA cannot accept all regular languages with bounded error.  It has been proved that the class of languages recognized by GQFA is stochastic languages with unbounded error \cite{46}. Hirvensalo \cite{42} introduced a generalized model of MO-1QFA, in which each transition function of each input symbol induces an entirely positive trace preserving mapping. It has been proved that MO-1gQFA and MM-1gQFA both recognize regular languages with bounded error \cite{43}.

Qiu et al. \cite{44} introduced a new computing model of 1QFA named one-way quantum finite automata together with classical states (1QFAC). Further, they proved that the class of languages recognized by 1QFAC is exactly all regular languages. They proved the equivalence problem of 1QFAC and showed that the minimization problem of 1QFAC is EXPSPACE. Therefore, 1QFAC is more powerful than MO-1QFA in terms of language recognition. In particular, 1QFAC is exponentially more concise than DFA for certain languages. Bianchi et al. \cite{102} investigated the comparison between QFA and CL-1QFA on basis of descriptional power and shown the procedure to convert CL-1QFA to corresponding DFA with at most exponential increase in size.

Yakaryilmaz \cite{22} conceptualized the idea of a blind counter with real-time quantum automata. Furthermore, he separated the real-time quantum automata with a blind counter (rtQ1BCA) with one-way deterministic k blind counter automata 1DkBCA by using the language $\it L^*$ , the Kleene closure of $L=\{a^nb^n \mid n\geq0\}$. It has been shown that $L=\{a^nb^n \mid n\geq0\}$ can be recognized by real-time quantum blind counter automata with negative one-sided error bound $\epsilon$. On the other hand, it has shown that $\it L$ can be easily recognized by one-way deterministic with one blind counter automata (1D1BCA). Yakaryilmaz \cite{22} also conjectured that one-way probabilistic $\it k$ blind counter automata (1PkBCA) cannot be designed for $\it L^*$. But recently, Nakanishi et al. \cite{48} disapproved this conjecture. They provide an algorithm for 1P1BCA that recognizes $\it L^*$ .

Amano and Iwama \cite{14} introduced the concept of 1.5QFA, demonstrating that the problem of emptiness was unsolvable using this model. 1.5QFA can accept some languages not otherwise recognized by 1QFA. Nakanishi et al. \cite{21} demonstrated that 1.5QFA could recognize the non-regular language $L=\{a^mdb^ncb^n \mid n\geq0, m\geq0\}$. They demonstrated that 1.5QFA could recognize non-context-free languages with a probability of less than 2/3, which can be recognized by QCPA with larger probability. Yakaryilmaz \cite{39} concluded that 1.5QFA can recognize non-stochastic languages.

Kondacs et al. \cite{10} introduced the concept of 2-way quantum finite automata (2QFA), demonstrating that it was more powerful than the classical two-way finite automata. In 2QFA, the R/W head can move either left or right, or may remain stationary. Furthermore, it has been shown that 2QFA can accept non-regular languages with a one-sided error in linear time. They also noted that 2QFA can accept non-context-free languages with a bounded error in linear time.  Ambainis et al. \cite{15} reported that 2QFA was difficult to implement because the quantum states needed to keep track of the position of the read/write head, which grows with the input string. Yakaryilmaz et al.\cite{23} defined a more efficient probability amplification technique for $L=\{a^nb^n \mid n\in N\}$. 2QFA that can accept the language They presented various techniques for reducing the state of probability amplification and designed machines with less state complexity. In recent years, the research exertion on QFA models has decisive on one-way models. However, the study of two-way quantum finite automata (2QFAs) is attracting less attention. Recently, Bhatia and Kumar \cite{108} modelled RNA secondary structures loops such as, internal loop and double helix loop using 2QFA. Bhatia and Kumar \cite{111} introduced a variant of 2QFA with multiheads and proved more powerful than its classical and quantum counterparts. It has been proved that two-way quantum multihead finite automata can be designed for a language containing of all words whose length is a prime number with less number of heads as compared to classical one.

Ambainis et al. \cite{24} described a concept of 2QFA with classical states: Two-way finite automata with quantum and classical states (2QCFA). They demonstrated that 2QCFA is more powerful (it can simulate any classical automaton and can recognize some languages that classical automata cannot) and can be implemented with a quantum part of constant size. They also proved that classical 2-way finite automata equipped with constant-size quantum registers can perform quantum transformations, measurements and can recognize non-regular language $L=\{a^nb^n \mid n\in N\}$. On the other hand, 1.5QFA and 2QFA have quantum registers whose size depends on the input length. Thus, they showed that 2QCFA is superior to these models that have finite amount of quantum registers. Furthermore, they designed 2QCFA for the palindrome language, which cannot be designed by two-way deterministic finite automata (2DFA).

Qiu \cite{33} demonstrated various Boolean operations over the class of languages recognized by 2QCFA. It has been proved that 2QCFA is closed under catenation under some conditions and also gave some examples of languages recognize by 2QCFA. Zheng \cite{25} proved that $L=\{xcy\mid x,y \in \{a,b\}^*, c \in \Sigma, \mid x \mid = \mid y \mid\}$   can be recognized by 2QCFA with one-side error probability in polynomial time, but, it can be recognized by 2PFA in exponential time with bounded error.

Zheng et al. \cite{47} introduced two-way two-tape finite automata with quantum and classical states (2TQCFA) model. They demonstrated efficient 2TFA (two-tape finite automata) algorithms to recognize languages which can be recognized by 2QCFA. Furthermore, they have given efficient 2TQCFA algorithms to recognize $L=\{a^nb^{n^2}\mid n\in N\}$. They have also introduced  $\it k$-tape automata with quantum and classical states (kTQCFA) and proved that it can recognize $L=\{a^nb^{n^k} \mid n\in N\}$ . Yakaryilmaz et al. \cite{26} defined a new model of two-way finite automata in which we can reset the R/W head to the leftmost position during computation. 

\subsection{Succinctness results}
In recent years, state complexity advantages of quantum finite automata is a new direction. The power of quantum finite automata solving promise problems is a hot topic. It has got enormous response from various researchers.  In 2009, Ambainis and Nahimovs \cite{79} designed the quantum automata using probabilistic argument. It has been shown that proposed automata use 4/{$\epsilon$} log 2\textit{p} states to recognize a $L_p=\{a^i |$ \text{\textit{i} is divsible by \textit{p}}$\}$, whereas classical automata take \textit{p} states. The exact quantum computation has been widely studied for promise problems.  Ambainis and Yakaryilmaz \cite{80} proved that infinite family of promise problems can be recognized by two-way quantum finite automata in realtime by just tuning transition amplitudes. It has been shown that $A_{yes}^k=\{a^{i2^{k}} |$ \text{\textit{i} is a non-negative integer}$\}$ and $A_{no}^k=\{a^{i2^{k}}|$ \text{\textit{i} is a positive odd integer}$\}$ can be recognized by two-way quantum finite automata in realtime, whereas any DFA takes at least $2^{k+1}$ states.

Bianchi et al. \cite{81} shown the superiority of quantum variant over DFA by considering the promise problems $(A^{N, r_1}_{yes}, A^{N, r_2}_{no})$ such that 0$\leq r_1 \neq r_2$, where $A^{N, r_1}_{yes}= \{\sigma^n | n \equiv r_1 ~mod~ N\}$  and $A^{N, r_2}_{no}= \{\sigma^n | n \equiv r_2 ~mod~ N\}$ over unary alphabet $\sigma$. It has been shown that extended version of promise problem can be recognized by MM-1QFA exactly, but DFA takes \textit{d} states, where  \textit{d} is the smallest integer  such that $d| N$ \text{and} $d\mid(r_2-r_1)~ mod~ N$. Rashid and Yakaryilmaz \cite{82} designed zero-error QFA for promise problems which cannot be recognized by PFA with bounded-error.  Yakaryilmaz and Say \cite{26} investigated various infinite families of regular languages which can be recognized by 2QFA with probability greater than 1/2
after tuning the transition amplitudes, whereas the size of one-way variants (1PFA and 1QFA) increase without bound. Zheng et al. \cite{83} examined the state succinctness of 2QCFA. It has been proved that $A_{yes}^{eq}=\{w=a^mb^m | w \in \{a,b\}^*\}$ and $A_{no}^{eq}=\{w \neq a^mb^m | w \in \{a,b\}^*, |w|\geq m\}$ $A_{yes}^{eq}=\{w=a^mb^m | w \in \{a,b\}^*\}$ can be recognized by a 2QCFA in a polynomial expected running time with  O(log 1/$\epsilon$) classical states and constant number of quantum states, whereas DFA takes at least at least 2\textit{m}+2 states and $\sqrt{m}$ for 2DFA and 2NFA respectively.

Gruska et al. \cite{84} generalized the distributed Deutsch–Jozsa promise problem to find the Hamming distance $H(x,y)$ between strings $x, y \in \{0, 1\}^n$ i.e weather $H(x,y)=n/2$ or $H(x, y)=$0. Further, the disjointness promise problem has been studied and shown an exponential gap between quantum (also probabilistic) and deterministic communication complexity.  Zheng et al. \cite{85} studied the state complexity of semi-quantum one way and two-way finite automata. It has been proved that promise problem $A_{eq}=\{A_{yes}(n), A_{no}(n)\}$, where and $A_{no}(n)=\{x\#y| x \neq y, \{x, y\}\in \{0,1\}^n\},$ and $A_{yes}(n)=\{x\#y| x = y, \{x, y\}\in \{0,1\}^n\}$ can be recognized by an exact 1QCFA with $O(n)$ classical states and \textit{n} quantum states, but the size of 1DFA is $2^{\Omega(n)}$. Further, the state complexity of 2QCFA is shown for languages $L(p)=\{a^{kp}|k \in \mathbb{Z}^{+}\}, p \in \mathbb{Z}^+$, and $C(m)=\{w|w \in \Sigma^m\}$ for any $m \in \mathbb{Z}^+$ and any input alphabet $\Sigma$. It has been proved that the language $C(m)$ can be recognized by 2QCFA in $O(1/\epsilon~ m^2 |w|^4$) expected running time with one-sided error $\epsilon$ using constant number of classical states and two quantum basis states, whereas 2PFA takes $\sqrt[3]{(log~ m)/b}$. Gruska et al. \cite{96} introduced two acceptance modes named recognizability and solvability, and explored various promise problems classical, finite, quantum and semi-quantum automata. It has been shown that quantum variants of classical automata have significant more power.

Zheng and Qiu \cite{86} presented a method for state succinctness results of QFA.  It has been shown that state succinctness outcomes can be extracted from query complexity results. The simpler quantum algorithm is given for partial function. It has been proved that promise problem $A(n)=\{A_{yes}(n), A_{no}(n)\}$, where and $A_{no}(n)=\{x\#y\#\#x\#y| \{x, y\}\in \{0,1\}^n, H(x, y)=n/2\},$ and $A_{yes}(n)=\{x\#y\#\#x\#y| \{x, y\}\in \{0,1\}^n, H(x, y)\in \{0, 1, n-1, n\}\}$ can be recognized  by 1QCFA with $O(n^3)$ classical states and $O(n^2)$ quantum states, but the size of 1DFA is $2^{\Omega(n)}$. Zheng et al. \cite{87} considered the concept of time-space complexity and proved that quantum computing models are more superior than classical variants in context of language recognition. The communication complexity results are used to derive time-space upper-bounds for 2QCFA and proved more desirable than probabilistic Turing machine (PTM). Further, it has been also proved that 2PFA is more superior than DTM in terms of time-space complexity. Bhatia and Kumar \cite{109} introduced the concept of quantum queue automata and proved more powerful than classical variants in real-time. Zheng \cite{94}  et al. examined the promise problems recognized by quantum, semi-quantum finite automata and classical automata. It has been proved that there exists a promise problem that can be recognized by MO-1QFA exactly, but cannot be recognized by DFA. Further, it has bee proved that there exists a promise problem that can be recognized by one-way finite automata with quantum and classical states (1QFAC), but 1PFA cannot be designed with error probability.

Zheng et al. \cite{92} proposed one-way finite automata with quantum and classical states (1QCFA) and examined its closure properties. Further, the main succinctness result is derived and the state complexity of 1QCFA is demonstrated. It has been proved that language ($L_m$) can be recognized by 1QCFA  with $O(log~m)$ quantum states and 12 classical states for any prime \textit{m}. But, $L_m$ is cannot be recognzied by MM-1QFA. Although, it can be recognized by any PFA with at last \textit{m} states. Bianchi \cite{100} shown that 1QCFA can recognize regular languages with isolated cut-point. Further, 1QCFA  is designed for word quotients and inverse homomorphic images of languages with isolated cut-point and polynomial increase in size.
From the perspective of state complexity, Qiu et al. \cite{93} studied one-way quantum finite automata with quantum and classical states (1QFAC) and stated that it is more concise than DFA 
exponentially. Further, Qiu et al. \cite{103} studied the equivalence problem of 1QFAC by a bilinearization technique. The quantum basis state minimization problem of 1QFAC can be solved in EXPSPACE. Further, there exists a polynomial-time $O((k_1n_1)^2+(k_2n_2)^2-1)^4$ algorithm to determine their equivalence. Recently, Gainutdinova and Yakaryilmaz \cite{97} investigated the computational power of PFA and QFA by studying promise problems based on unary languages. It has been proved that there exists unary promise problems which can be recognized by QFAs, but cannot be recognized by 
PFAs with bounded error. Further, they have considered a promise problem with two parameters and  shown that QFA is more succinct than PFA by fixing one parameter, whereas PFA
can be more succinct than DFA exponentially after fixing the other parameter.

\subsection{Other results}
Interactive proof (IP) systems with verifiers are modeled by quantum finite automata, where a powerful quantum prover communicates with a quantum-automaton verifier via a shared communication channel. The quantum prover is computationally strong (unlimited), whereas the power of verifier is limited. In 2009, Nishimura and Yamakami \cite{35} explored a direct application of QFA to IP, where QFA is a verifier and prover can be unitary operation. The computational power of MO-1QFA, 1QFA and 2QFA is studied by imposing restrictions. QIP(\textit{A}) represents the class of languages recognized by QIP associated with verifier \textit{A}. It has been shown that the language recognition power of 1QFA is increased with an interaction of prover i.e. QIP(1QFA) is equal to regular languages. Further, it has been proved that 2QFA a is proper subset of QIP(2QFA) in polynomial time for language $P=\{x\#x^R|x \in\{0, 1\}^*\}$, where \# is a separator. Furthermore, It has been demonstrated that the QIP(MO-1QFA) verifier is not closed under complementation and QIP(2QFA) is closed under union. Zheng et al. \cite{88} investigated the power of QIP associated with 2QCFA verifier. It has been shown that the languages $L_{m}=\{xay|x, y \in \{a, b\}^*,|x|=|y|\}$ and $L_{p}=\{xax^R|x, \in \{a, b\}^*\}$ can be recognized by QIP(2QFA) in polynomial and exponential time with one-sided bounded error.

A debate system is a generalized version of IP system, where two provers argue over
the belongingness of particular string (\textit{x}) in a language (\textit{L}). The prover prompts the verifier that $x \in L$ and refuter attempts to prove that \textit{x} does not belong to \textit{L}. Yakaryilmaz et al. \cite{68} investigate that quantum model surpass the classical ones when bounded to run in polynomial time. It has been proved that non-context free language $L^P=\{1^p|p~ \text{is prime}\}$ and $L^m=\{1^{2^m}|m>0\}$ can be recognized by 2QCFA in polynomial time debates using two qubits, $L^s=\{1^{m^2}|m>0\}$ and $L^n=\{1^n|n~\text{is a fibnacci number}\}$ with three qubits respectively. In 2015, Nishimura and Yamakami \cite{90} demonstrated the strengths and weaknesses of QIP(QFA) by imposing some restrictions on verifiers and provers behavior in context of language recognition power. Yamakami \cite{98} introduced QFA with extra information known as advice, which depends on the length of an input string. It has explored that QFA with advice cannot be designed for some regular languages.  Yamakami \cite{91} extended the single prover model, where more than one prover can interact with a verifier named quantum multiple prover
interactive proof (QMIP) systems and shown the advantages of quantum computation over classical variants. It has been proved that QIP(1QFA) $\neq$ QMIP(1QFA) and QIP(2QFA) $\neq$ QMIP(2QFA) in polynimal time. Scegulnaja-Dubrovska et al. \cite{105} shown that palindrome language can be recognized by MM-1QFA with postselection, whereas it cannot be recognized by PFA with non-isolated cut-point 0.  Yakaryilmaz and Say \cite{106} proposed a quantum finite automata with postselection and proved that the computational power of real-time probabilistic and quantum finite automata can be increased with postselection. Further, the class of languages recognized by QFA with postselection are examined.

\begin{center}
	\begin{longtable}{p{2.3cm}|p{13cm}} 
		\caption{Notable finding of various QFA models}\\ 
		\hline \multicolumn{1}{c}{\textbf{Authors}} & \multicolumn{1}{c}{\textbf{Notable Findings}} \\ \hline 
		\endfirsthead
		
		\multicolumn{2}{c}%
		{{\bfseries \tablename \thetable{} -- continued from previous page}} \\
		\hline \multicolumn{1}{c}{\textbf{Authors}} &
		\multicolumn{1}{c}{\textbf{Notable Findings}}  \\
		\hline 
		\endhead
	
		\endfoot	
		\hline
		\endlastfoot
		Moore et al., 2000 \cite{9} & \begin{itemize}
				\vspace{-0.2cm}
			\item Introduced the concept of QFA and QPDA.
			\vspace{-0.2cm}
			\item QFA and QPDA can be designed for quantum regular and quantum context-free language respectively.	
			\vspace{-0.2cm}
			\item They analyzed various closure properties, pumping lemmas and rational generating functions for QFA. 
			\vspace{-0.2cm}
			\item Carry out a comparative study of quantum and classical automata. 
		\end{itemize}\\
	\hline
	Ambainis et al., 2000 \cite{16} &  \begin{itemize}
			\vspace{-0.2cm}
		\item	Characterized the class of languages recognized by 1QFA.
		\vspace{-0.2cm}
		\item	Examined the necessary and sufficient condition for a language to be recognized by 1QFA.
		\vspace{-0.2cm}
		\item	Proved that 1QFA is not closed under union and Boolean operation.
	\end{itemize}\\ 
	\hline
	Brodsky et al., 2008 \cite{17}	& \begin{itemize}
		\vspace{-0.2cm}
		\item	Characterize the MO-1QFA on the basis of bounded error acceptance mode.
		\vspace{-0.2cm}
		\item	Simulate the MO-1QFA by PFA. 
		\vspace{-0.2cm}
		\item	MO-1QFA restricted in term of computational power, and it can accept only group languages.
	\end{itemize}\\
	\hline
	Ambainis et al., 1998 \cite{15} & \begin{itemize}
			\vspace{-0.2cm}
		\item Carry out a comparative study between 1QFA and classical automata.
		\vspace{-0.2cm}
		\item Simulate 1QFA by reversible finite automata for the recognition of languages. In fact, they showed that if we allow smaller probabilities, MM-1QFA can be more powerful than 1RFA.
		\vspace{-0.2cm}
		\item Proved that 1QFAs are space-efficient.
		\vspace{-0.2cm}
		\item Proved that if a reversible finite automaton can be designed for a language, and then 1QFA can be designed that accepts the same language with probability more than 7/9.
	\end{itemize}\\
	\hline
	Kikusts, 1998 \cite{19}	& \begin{itemize}
		\vspace{-0.2cm}
		\item Compared the MM-1QFA with the classical automata.
		\vspace{-0.2cm}
		\item Designed an MM-1QFA, which is quadratically smaller in size than its corresponding deterministic finite automaton.
	\end{itemize}\\
	\hline
	Ambainis et al., 1999 \cite{18}	& \begin{itemize}
		\vspace{-0.2cm}
		\item	Introduced a technique for encoding the classical bits into fewer quantum bits using the principle of entropy coalescence.
		\vspace{-0.2cm}
		\item	Proved that MM-1QFA has exponentially more states than its nominal DFA for a particular language.
	\end{itemize}\\
\hline
	Brodsky et al., 2002 \cite{17} & \begin{itemize}
		\vspace{-0.2cm}
		\item	Proved that MM-1QFA is closed under complement, inverse homomorphisms and word quotients.
		\vspace{-0.2cm}
		\item	Proved that MM-1QFA is not closed under homomorphisms.
		\vspace{-0.2cm}
		\item	They have given the necessary and sufficient condition for a class of languages recognized by MM-1QFA.
	\end{itemize}\\
	\hline 
 & \multicolumn{1}{c}{\textbf{LQFA}} \\
	\hline
	Ambainis et al., 2004 \cite{20}	& \begin{itemize}
		\vspace{-0.2cm}
		\item	Proposed the concept of LQFA.
		\vspace{-0.2cm}
		\item	Carry out a comparative study of LQFA with Brodsky and Pippenger's MO-1QFA and MM-1QFA. 	
		\vspace{-0.2cm}
		\item	Proved that LQFA are closed under complement, union, word quotients and inverse homomorphism.
		\vspace{-0.2cm}
		\item	Proved that LQFA can be designed for a proper subset of the languages recognized by MM-1QFA with bounded error.
	\end{itemize}\\
	\hline
& \multicolumn{1}{c}{\textbf{CL-1QFA}} \\
	\hline
	Bertoni et al. 2003 \cite{30}	& \begin{itemize}
		\vspace{-0.2cm}
		\item	Proposed a quantum computing model named quantum finite automata with control states.
		\vspace{-0.2cm}
		\item	Proved that the class of languages recognized by CL-1QFA with an isolated cut point is closed under Boolean operations.
		\vspace{-0.2cm}
		\item	Proved that the languages recognized by CL-1QFA are regular languages with bounded error.
		\vspace{-0.2cm}
		\item	They demonstrated that  -state CL-1QFA could be transformed into   BLM (Bilinear machine), where  denotes the number of states of a minimized deterministic finite automata.
	\end{itemize}\\
	\hline
	& \multicolumn{1}{c}{\textbf{MO-1gQFA} }\\
	\hline
	Hirvensalo 2010 \cite{42}	& \begin{itemize}
		\vspace{-0.2cm}
		\item Introduced a model in which each transition function of each input symbol induces a completely positive trace preserving mapping.
		\vspace{-0.2cm}
		\item	Proved that the class of languages recognized by MO-1gQFA is regular languages with bounded error. 
		\vspace{-0.2cm}
		\item	It can also simulate classical DFA and even probabilistic automata.
	\end{itemize}\\
	\hline
	& \multicolumn{1}{c}{\textbf{MM-1gQFA} }\\
	\hline
	Li et al. 2012 \cite{43}	& \begin{itemize}
		\vspace{-0.2cm}
		\item Studied MM-1gQFA from the language recognition power and the equivalence problem. 
		\vspace{-0.2cm}
		\item	Proved that the class of languages recognized by MM-1gQFA is regular languages with bounded error.
		\vspace{-0.2cm}
		\item 	Proved the equivalence problem of MM-1gQFA.
	\end{itemize}\\
	\hline
	& \multicolumn{1}{c}{\textbf{1QFAC} }\\
	\hline
	Qiu et al. 2011 \cite{44} & \begin{itemize}
		\vspace{-0.2cm}
		\item Introduced a new computing model of 1QFA named one-way quantum finite automata together with classical states (1QFAC).
		\vspace{-0.2cm}
		\item	Proved that the class of languages is recognized by 1QFAC are exactly all regular languages.
		\vspace{-0.2cm} 
		\item	It is more powerful than MO-1QFA in terms of language recognition. 
		\vspace{-0.2cm}
		\item	In particular, 1QFAC is exponentially more concise than DFA for certain languages.
		\vspace{-0.2cm}
		\item	Proved the equivalence problem of 1QFAC.
	\end{itemize}\\
	\hline
& \multicolumn{1}{c}{\textbf{2QFA} }\\
	\hline
	Kondacs et al., 1997 \cite{10}	& \begin{itemize}
		\vspace{-0.2cm}
		\item	Introduced the concept of 2QFA.
		\vspace{-0.2cm}
		\item	2QFA is a quantum version of 2-way deterministic finite automata.
		\vspace{-0.2cm}
		\item They proved that 2QFA could be designed for non-regular languages with one-sided error.
		\vspace{-0.2cm}
		\item	They proved that 2QFA could be designed for non-context-free languages with a bounded error in linear time.
	\end{itemize}\\
	\hline
	Ambainis et al., 2002 \cite{24} & \begin{itemize}
		\vspace{-0.2cm}
		\item	Introduced a new quantum model 2QCFA. 
		\vspace{-0.2cm}
		\item 2QCFA is identical to 2QFA with classical states.
		\vspace{-0.2cm}
		\item	They proved that 2QCFA model is more powerful than the classical two-way finite automata and can be implemented with a quantum part of constant size.
		\vspace{-0.2cm}
		\item	They proved that classical 2-way finite automata equipped with constant-size quantum registers can perform quantum transformations, measurements and can recognize non-regular language  $L=\{a^nb^n \mid n \in N\}$.  
		\vspace{-0.2cm}
		\item	They have designed 2QCFA for palindromes which are not possible with 2-way deterministic finite automata.
	\end{itemize}\\
	\hline
	Qiu, 2008 \cite{33} & \begin{itemize}
		\vspace{-0.2cm}
		\item	He proved that 2QCFA were closed under concatenation operation under certain conditions.
		\vspace{-0.2cm}
		\item	Given the language  $L=\{xx^R \mid x \in \{a,b\}^*, \#x(a)=\#x(b)\}$, where  denotes the number of ’s in . 2QCFA can be designed for language L with one-side error probability in polynomial time. 
		\vspace{-0.2cm}
	\end{itemize}\\
	\hline
& \multicolumn{1}{c}{\textbf{Recent results} }\\
	\hline
	Gainutdinova and Abuzer Yakaryılmaz 2015 \cite{64} & \begin{itemize}
		\vspace{-0.2cm}
		\item Investigated the power of QFA and PFA on unary promise problems. Proved that QFAs are more powerful than PFAs with bounded error acceptance.
		\vspace{-0.2cm}		
		\item	Shown that the computational power of Las-Vegas QFAs and bounded-error PFA is equivalent to DFAs on binary problems.
		\vspace{-0.2cm}
		\item	Further, it has been investigated that on fixing one parameter of a QFA, it becomes more succinct than PFA for two parameters unary promise problems.
	\end{itemize}\\
	\hline
	Demirci et al. 2014 \cite{65}	& \begin{itemize}
		\vspace{-0.2cm}
		\item Presented the decidability results of real-time classical and quantum alternating models.
		\vspace{-0.2cm}
		\item	Proved that alternating QFAs on unary alphabets are undecidable with two alternations. But, it is decidable for nondeterministic QFAs on general alphabets.
		\vspace{-0.2cm}
		\item	Shown that unary squares language can be recognized by alternating QFAs with two alternations.
		\vspace{-0.2cm}
		\item	On the other hand, they have defined real-time private alternating finite automata (PAFA) and proved that it can recognize some non-regular unary languages, and its emptiness problem is undecidable.
	\end{itemize}\\
	\hline
	Nakanishi et al. 2014 \cite{48}	& \begin{itemize}
		\vspace{-0.2cm}
		\item Investigated that one-way quantum one-counter automata (1Q1CA) is more powerful than its probabilistic counterpart on promise problems with zero-error.
		\vspace{-0.2cm}
		\item	Shown that promise problem XOR-EQ can be solved by 1Q1CA exactly, which cannot be solved by any one-way deterministic one-counter automata (1D1CA). Nakanishi et al. \cite{48} provide an algorithm for 1P1BCA that recognizes $\it L^*$ to disapprove the conjecture by Yakaryilmaz \cite{22}.
	\end{itemize}\\
	\hline
	Say and Yakaryılmaz 2014 \cite{66}	& \begin{itemize}
		\vspace{-0.2cm}	
		\item Introduced a modern quantum finite automaton involving superoperators. 
		\vspace{-0.2cm}
		\item Shown that modern QFA can outperform various classical counterparts.
		\vspace{-0.2cm}
		\item	Proved that non-regular language $L_{EQ}=\{w \mid w \in \{a,b\}^*, w_a=w_b\}$ cannot be recognized by 2PFA with bounded error in polynomial time, but it can be recognized by 2QFA.
		\vspace{-0.2cm}
		\item	Shown that any real-time PFA can be simulated by a real-time QFA having the same number of states.
	\end{itemize}\\
	\hline
	Say and Yakaryılmaz 2014 \cite{67} & \begin{itemize}
		\vspace{-0.2cm}
		\item	Given a constant space theorem which states that: for every language $\it L$ there exists a language $\it L^{'}$ on the alphabet $\{a,b\}$ such that $\it L^{'}$ is Turing equivalent to $\it L$ and there exists a 2QCFA which recognizes $\it L^{'}$ with bounded error in polynomial time.
		\vspace{-0.2cm}
		\item	Constructed a public-coin interactive proof system, where messages are classical and verifier is 2QCFA.
		\vspace{-0.2cm}
		\item	Investigated that when 2QCFA is used as verifiers in public-coin interactive proof systems, then it can verify membership in all languages with bounded error and outperform classical counterparts.
		\vspace{-0.2cm}
		\item	Proved that for any language \textit{L} on the binary alphabet $B=\{a,b\}$   there exists a bounded error public-coin interactive proof system, where the verifier is a 2QCFA with two quantum bits, and the expected runtime is computed as $2^{2^{O(n)}}$, where $\it n$ is the length of input.
		\end{itemize}\\
		\hline
		Yakaryılmaz et al 2014 \cite{68} & \begin{itemize}
			\vspace{-0.2cm}	
			\item	Studied a model where the verifier interacts with a prover who tries to convince the verifier that input string $x\in L$ and a refuter tries to prove that $x \notin L$. 
			\vspace{-0.2cm}
			\item	Proved that languages $L_{UPRIME}=\{1p\mid p ~ is ~ a ~ prime\}$ and $L_{POWER}=\{1^{2^m} \mid m>0\}$ has polynomial time debates by 2QCFA as a verifier with only two qubits.
			\vspace{-0.2cm}
			\item	Proved that every Turing-decidable language has debates checkable by a four quantum states 2QCFA with bounded error, by considering only rational entries in its quantum operators.
			\vspace{-0.2cm}
			\item	Demonstrated that some non-context free languages have short debates with 2QCFA as quantum verifiers and proved that it can outperform its classical counterpart when constrained to run in polynomial time.
		\end{itemize}\\
		\hline
		Giannakis et al. 2015 \cite{70}	& \begin{itemize}
			\vspace{-0.2cm}	
			\item	Proposed quantum version of  $\omega$-automata for infinite periodic words.
			\vspace{-0.2cm}
			\item	Defined simple periodic and simple  $\it m$-periodic quantum  $\omega$-automata.
			\vspace{-0.2cm}
			\item	Proved that it can recognize $(a^mb)^\omega$, where $it m$ is a finite number space efficiently.
		\end{itemize}\\
		\hline
		Belovs et al. 2016 \cite{72} & \begin{itemize}
			\vspace{-0.2cm}	 \item	Determined that 2QFA with minimum number of states can separate a pair of words.
			\vspace{-0.2cm}
			\item	It has been investigated that a language which cannot 	Shown that 2QFA can separate any easy pair with zero-error but cannot separate some hard pairs even in nondeterministic acceptance mode by using real amplitudes. 
			\vspace{-0.2cm}
			\item	Proved that 2QFA with complex amplitudes can separate any pair in nondeterministic acceptance mode.
		\end{itemize}\\
		\hline
		Ganguly et al. 2016 \cite{74}	& \begin{itemize}
			\vspace{-0.2cm}	
			\item	Proposed one-way multihead quantum finite automata (k-1QFA) and proved that it can recognize all unary languages. 
			\vspace{-0.2cm}
			\item Shown that a language $L=\{w=w^R \mid w \in \{a,b\}^*\}$ cannot be recognized by 1-way deterministic 2-head finite automata (2-1DFA), but can be recognized by 2-1QFA.
			\vspace{-0.2cm}	
			\item	Furthermore, it has been investigated that it is more powerful than 1-way reversible 2-head finite automata.
		\end{itemize}\\
		\hline
		Ganguly et al. 2016 \cite{75} & \begin{itemize}
			\vspace{-0.2cm}	
			\item	Introduced two-tape one-way quantum finite automata (2-2T1QFA) and claimed that it can recognize all regular languages. 
			\vspace{-0.2cm}
			\item	It has been investigated that a language which cannot be recognize by any deterministic multihead finite automata can be recognized by 2-2T1QFA. 
		\end{itemize}\\
		\hline
		Bianchi et al. 2017 \cite{73} & \begin{itemize}
			\vspace{-0.2cm}	
			\item	Outlined the Bertoni’s ideas related quantum computational theory.
			\vspace{-0.2cm} 
			\item	Introduced the more advanced statistical framework to prove the existence of small quantum finite automata accepting periodic languages.
			\vspace{-0.2cm} 
			\item	Shown the promise problems to relate the power of QFA with their classical counterparts by considering multiperiodic languages.  
		\end{itemize}\\
		\hline
		Bhatia and Kumar 2018 \cite{77} \cite{107}, Saggi \cite{110} & \begin{itemize}
			\vspace{-0.2cm}	
			\item	Investigated the relation of quantum finite-state machines (QFSM) with matrix product state (MPS) of quantum spin systems.
			\vspace{-0.2cm}
			\item	Efficiently simulated MPS (GHZ state, AKLT state, Cluster state and W state) with a broader quantum computational theory using unitary criteria. 
			\vspace{-0.2cm}
			\item	Proved that QFSM is equivalent to MPS representations of ground state of quantum spin systems.
			\item Simulated MPS on a quantum computer using circuits and the probability distribution among the quantum states is calculated.
		\end{itemize}\\
		\hline
		Khrennikov and Yurova 2017 \cite{78} & \begin{itemize}
			\vspace{-0.2cm}	
			\item	Explored the interrelation between dynamics of conformational and functional states of proteins by quantum phenomena.
			\vspace{-0.2cm}
			\item	Proposed a model to analyze protein behaviour by using concepts of automata theory and shown the similarity between modelling of behaviour of proteins and quantum systems.
		\end{itemize}\\
		\hline
	\end{longtable} 
\end{center}

\subsection{Closure Properties}
\begin{table} [h]
	\centering
	\caption{Closure properties of various QFA models, Here, \cmark, \xmark, - represents that particular model is closed, not closed and undefined respectively.}
	\begin{tabular}{ |p{1.8cm}|p{1.5cm}|p{1.8cm}|p{2.1cm}|p{1.2cm}|p{2cm}|p{1.5cm}|p{1.4cm}|}
		\hline
		\textbf{Authors}& \textbf{Models}& \textbf{Homomo-}
		\textbf{rphism}& \textbf{Inverse homomorphism}& \textbf{Union}& \textbf{Word quotients}& \textbf{Boolean operations}& \textbf{Comp-}
		\textbf{lement}
		\\ 
		\hline
		Brodsky et al. \cite{17} & MO-1QFA & \xmark & \cmark & \cmark & $-$ & \cmark & $-$ \\
		\hline
		Ambainis et al.\cite{16}
		&	MM-1QFA  & \xmark & \cmark & \cmark & $-$ & $-$ & \cmark \\
		\hline
		Ambainis et al. \cite{20} &	LQFA  & $-$ & \cmark & \cmark & \cmark & \cmark & \cmark \\
		\hline
		Bertoni et al. \cite{30} &	CL-1QFA   & $-$ & $-$ & $-$ & $-$ & \cmark & $-$ \\
		\hline
		Macko \cite{27} &	1.5QFA    & \xmark & $-$ & $-$ & $-$ & $-$ & $-$ \\
		\hline
		Macko \cite{27} &	2QFA    & \xmark & $-$ & $-$ & $-$ & $-$ & $-$ \\
		\hline
		Qiu \cite{33} &	2QCFA    & $-$ & $-$ & $-$ & $-$ & \cmark & \cmark \\
		\hline
	\end{tabular}
\end{table}
	
	In this subsection, we describe a comparative study of various quantum finite automata based on their closure properties. Macko \cite{27} demonstrated that the class of languages accepted by 1.5-way and two-way quantum finite automata are closed under non-erasing inverse homomorphism and inverse length non-increasing homomorphism. Qiu \cite{33} proved that 2QCFA is closed under intersection, complement and reversal operation with error probabilities. Further, it is closed under catenation operation with certain restricted condition. Table 8 illustrates a comparative study of various quantum models based on their closure properties.
	
	\subsection{Equivalence and Minimization of quantum finite automata}
	In quantum automata theory, checking the equivalence of two quantum finite automata is a challenging task. Two quantum finite automata are said to be equivalent if both accept any input string $\it w$ with the same probability. The equivalence problems inherent to MO-1QFA, MM-1QFA, and CL-1QFA have been explored by various researchers \cite{28,29,30,31,32}. Mateus et al. \cite{49} demonstrated that minimization of a given 1QFA with algebraic numbers is decidable and proposed an algorithm that takes automata as input and produce a minimal size equivalent automata.  This algorithm runs in an exponential space (EXPSPACE). 
	\begin{figure}[h]
		\centering
		\includegraphics[scale=0.35]{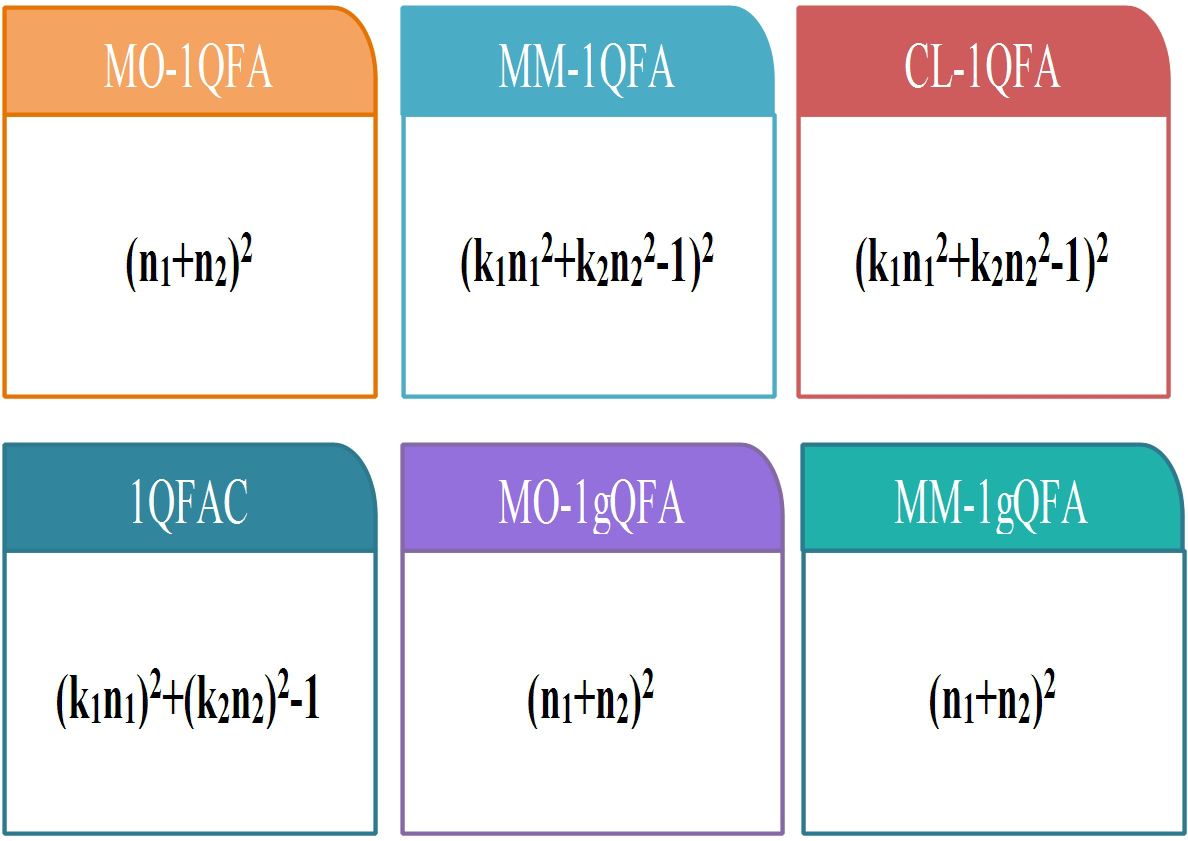}
		\caption{Summary of various quantum finite automata models}
	\end{figure}
	
	Moore et al. \cite{9} demonstrated that $\it n$-dimensional generalized QFA could be simulated by $\it n^2$-dimensional generalized stochastic automata. Koshiba \cite{28} proposed a polynomial-time algorithm for checking the equivalence of 1QFA. Li et al. \cite{29} demonstrated that MO-1QFAs $\it L$ and $\it M$, with $\it n_1$ and $\it n_2$ states respectively, are equivalent if and only if they are  $(n_1+n_2)^2$-equivalent. The state minimization of MO-1QFA is decidable in EXPSPACE \cite{49}. Bertoni et al. \cite{30} introduced CL-1QFA as a variant of the MM-1QFA model, in which a set of possible results during measurement are fixed. They demonstrated that $\it n$-state CL-1QFA could be transformed into $\it kn^2$ BLM (Bilinear machine), where $\it k$ denotes the number of states of a minimized deterministic finite automata. Li et al. \cite{31} demonstrated that BLMs $\it L_1$ and $\it L_2$, with $\it n_1$ and $\it n_2$.
	
	Due to their intricate nature, CL-1QFA and MM-1QFA cannot be converted into bilinear machines. Therefore, for solving the equivalence problem of MM-1QFA, they are transformed into CL-1QFA. Similarly, Li et al. \cite{32} demonstrated that two MM-1QFAs were equivalent if and only if they are $(k_1n_1^2+k_2n_2^2-1)$-equivalent, where $\it n_1$ and $\it n_1$ are the number of states of MM-1QFAs, and $\it k_1$ and $\it k_1$ are the number of states of minimal DFA. As an application of MM-1QFA equivalence, they demonstrated that two MM-1QFAs were equivalent if and only if they are $(3_1n_1^2+3_2n_2^2-1)$-equivalent, where 3 is the number of states of DFA to recognize a regular language $L=g^*a(a\mid r\mid g)^*$. Furthermore, there is a polynomial-time algorithm that takes input from two MM-1QFAs and decides whether they are equivalent or not in $O((3n_1^2+3n_2^2)^4)$  time \cite{32}. Using the algorithm proposed by Mateus et al. \cite{49}, the state minimization of MM-1QFA is also decidable in EXPSPACE.
	
	Qiu et al.\cite{44} proved that any two 1QFAC $\it A_1$ and $\it A_2$ are equivalent iff they are $(k_1n_1)^2+(k_2n_2)^2-1$ equivalent, where $\it k_1, k_2$ are number of classical states and $\it n_1, n_2$  are the number of quantum states of $\it A_1, A_2$ respectively by a bilinearization technique. The minimization of quantum basis states for 1QFAC is proved to be decidable in EXPSPACE \cite{49}. Further, there exists a polynomial-time $O((k_1n_1)^2+(k_2n_2)^2-1)^4$ algorithm to determine their equivalence. Li et al. \cite{43} demonstrated that two MO-1gQFA’s and MM-1gQFA’s on the same input alphabet $\Sigma$ are equivalent if and only if they are $(n_1+n_2)^2$-equivalent, where $n_1=dim(H_i)$ for $i=1,2,.$  and $n_1,n_2$ is dimension of spanned Hilbert space respectively. Similarly, it can be easily proved that the states minimization of MO-1gQFA and MM-1gQFA is decidable in EXPSPACE \cite{49}. Qiu et al. \cite{101} proved that any multi-letter QFAs $\it A_1$ and $\it A_2$ are equivalent iff they are $(n^2m^{k-1}-m^{k-1}+k)$-equivalent, where $m=|\Sigma|, n=n_1+n_2, k=max(k_1, k_2)$, $n_1 , n_2$ are the number of states of $A_1$ and $A_2$ respectively. Further, there exists a polynomial-time $O(m^{2k-1}n^8+km^kn^6)$ algorithm to determine their equivalence.

\subsection{Simulations of classical counterparts}
It is known that computation process of most classical models is either deterministic or probabilistic. Therefore, investigate the power of PFA to their quantum counterparts is a natural goal. Till now, there are several QFA models to simulate the classical automata models exactly. 
\begin{thm}
Let $\it L$ be a language accepted by MO-1QFA with cut-point $\lambda$ then there exists a PFA accepting same language with cut-point $\lambda'$.  
\end{thm}
The proof of the theorem is given in \cite{17,42,60}. It is also valid for probabilistic and quantum Turing machines \cite{61}. Therefore, it can be easily checked that any PFA can be converted into an equivalent QFA model.  
\begin{thm}
Let $\it L$ be a language accepted by $\it n$-state 1QFA, then there exists a $\it n^2$-state GFA accepting same language.
\end{thm}
\begin{proof}
The proof has been shown in \cite{60,32}.
\end{proof}
\begin{thm}
If any language $\it L$ is recognized by $\it n$-state 1QFA (pure states) with bounded error, then it can be recognized by $2^{O(n)}$-states 1DFA. 
\end{thm}
\begin{proof}
The proof of the theorem is given in \cite{10,15}.
\end{proof}
There exists a more powerful generalization of 1QFA model that allow mixed states named GQFA. Any 1DFA can be easily simulated by a GQFA
\begin{thm}
For every regular language recognized by PFA, there exists a MO-1gQFA recognizing it with certainty. 
\end{thm}
\begin{proof}
DFA is a special probabilistic automata which can be simulated exactly by an MO-1gQFA. The proof of the simulation process is shown in \cite{42}. However, the detailed proof is given by Li et al.\cite{43}. 
\end{proof}
\begin{thm}
If any language $\it L$ is recognized by $\it n$-state 1QFA (mixed states) with bounded error, then it can be recognized by $2^{O(n^2)}$-states 1DFA. 
\end{thm}
\begin{proof}
Simulation process of 1DFA by MO-1gQFA has been shown in various papers. The simulation by considering upper bound on the number of states is given in \cite{51}. 
\end{proof}

\section{Some Open Problems}
Based on literature survey, we would like to suggest some open problems for further consideration \cite{9,10,11,12,16,43,44,45,46,47}. 
\begin{longtable}{p{14cm}}
	\hline
	\centering \textbf{Open problems of 1QFA}
	\label{xyz}
	\endfirsthead
	\endhead
	\hline
	\textbullet	~To characterize the class of languages accepted by MM-1QFA with bounded error acceptance. \\
	\textbullet	~Investigate the relation between the state complexities of multi-letter QFAs with MO-1QFAs for unary regular languages.\\
	\textbullet	~To determine the power of MO-1QFA with unbounded error acceptance mode. \\
	\textbullet	~To determine a regular language that cannot be accepted by 1QFA but can be accepted by quantum counter automata. \\
	\textbullet	~To determine whether the class of languages recognized by MM-1QFA are closed under Boolean operations or not. \\
	\hline
	 \multicolumn{1}{c}{\textbf{Open problems of 1.5 QFA} }\\
	\hline 
	\textbullet	~To determine whether 1.5 QFA can recognize regular languages with bounded error $~ ~ ~ ~$ acceptance. \\
	\textbullet	~To carry state complexity comparative study of 1QFA and 1.5 QFA. \\
	\textbullet	~To determine the various closure properties of 1.5 QFA. \\
	\hline
 \multicolumn{1}{c}{\textbf{Open problems of 2QFA} }\\
	\hline 
	\textbullet	~To investigate a language for which 2QA with bounded error takes polynomial time, whereas 2PFA with a bounded error of the same languages take exponential time. \\
	\textbullet ~To determine the decidability of 2QFA and 1.5 QFA equivalence. \\
	\textbullet ~To determine whether there exist a non-stochastic language which can be represented by 2QFA but not by 1.5-way QFA? \\
	\textbullet	~To find out whether 2QFA can accept any non-stochastic languages with bounded error mode. \\
	\textbullet	To determine the various closure properties of 2QFA. \\
	\textbullet ~To determine whether 2QFA is more powerful than corresponding classical automata if it is restricted to a particular measurement after a specified time.\\
	\hline
	 \multicolumn{1}{c}{\textbf{Miscellaneous models}}\\
	\hline 
	\textbullet	~To carry state complexity of 1QFAC and carry out comparison with 1QFA models. \\
	\textbullet	~To determine the simulation of 1QFAC by 1QFA models. \\
	\textbullet	~To investigate whether GQFA can recognize strictly more languages than MM-1QFA with an unbounded error. \\
	\textbullet	~To determine whether any two 2QCFA are equivalent. Does there exist a polynomial-time algorithm to determine equivalence between them?\\
	\hline
	\multicolumn{1}{c}{\textbf{Other promising research directions}}\\
	\hline 
	\textbullet	~To study the computational power of QFA with
	algebraic methods. \\
	\textbullet	~To investigate the power of other QFA models with advice from computational complexity point of view.  \\
	\textbullet	~To explore more languages for which QFA can be designed with less number of states as compared to PFA and DFA. \\
	\textbullet	~To determine more promise problems to show separations between computational models.\\
	\hline
\end{longtable}

\section{Conclusion}

A quantum finite automaton is a theoretical model with finite memory which lay down the vision of quantum processor. There are various quantum automata models which are more powerful than classical ones. In this paper, we have comprehensively reviewed and analyzed various aspects of quantum finite automata based on past research literature. It helps to understand theoretically the fundamentals of various quantum computing models. We have outlined their definitions, behavior, closure properties, language recognition power, comparison, inclusion relationships, equivalence criteria, minimization and simulation results. We have subsequently summarized the literature published to date in the form of a systematic evolution of QFA models. In this review, we have compared various quantum finite automata models. This study makes a positive contribution to the growing body of quantum computing literature by exploring the computational power of various QFA models. Moreover, we have recognized and addressed various issues present in the research and identified some outstanding research questions still unresolved in various QFA models. Furthermore, various open problems are identified as a future area of research in the field of quantum automata theory. 

\begin{longtable}{p{2.5cm} p{12.3cm}}
\hline
\textbf{Appendix 1:} & \textbf{Acronyms}\\ 
\hline
QFT	& Quantum Fourier Transform\\
1QFA &		One-way Quantum finite automata\\
DFA	& Deterministic finite automata\\
PFA	 & Probabilistic finite automata\\
2PFA &	Two-way probabilistic finite automata\\
MO-1QFA	& Measure-once quantum finite automata\\
MM-1QFA	& Measure-many quantum finite automata\\
LQFA	& Latvian quantum finite automata\\
CL-1QFA	& Quantum finite automata with control languages\\
1QFAC	& One-way finite automata with quantum and classical states.\\
GQFA	& General quantum finite automata\\
MO-1gQFA & Measure-once general quantum finite automata\\
MM-1gQFA  &  Measure-many general quantum finite automata\\
QCPA	& Quantum pushdown automata with the classical stack\\
1.5 QFA   &   1.5 quantum finite automata\\
2QFA	& Two-way quantum finite automata\\
2QCFA	& Two-way finite automata with quantum and classical states\\
rtQ1BCA	& Real-time quantum automata with one blind counter\\
1DkBCA	 & One-way deterministic automata with k blind counter\\
1D1BCA	& One-way deterministic automata with one  blind counter\\
1Pk1BCA	& One-way probabilistic automata with k blind counter\\
1P1BCA	&	One-way probabilistic automata with one blind counter\\
2TFA 	& Two-tape finite automata\\
kTQCFA	 & $\it k$-tape automata with quantum and classical states\\ 
2TQCFA	& Two-way two tape finite automata with quantum and classical states\\
PAFA	& Real-time private alternating finite automata\\
1Q1CA	& One-way quantum one-counter automata\\
2-1DFA	& One-way deterministic two-head finite automata\\
2-2T1QFA &	Two-tape one-way quantum finite automata with two-head\\
QFSM &	Quantum finite-state machine\\

\hline

\end{longtable}

\section*{Acknowledgments}
Amandeep Singh Bhatia was supported by Maulana Azad National Fellowship (MANF), funded by Ministry of Minority Affairs, Government of India. 

\label{sec:test}
\bibliographystyle{elsarticle-num}
\bibliography{sample}

\end{document}